\documentclass[journal]{IEEEtran}
\usepackage{amsmath,amssymb,amsfonts}
\usepackage{amsthm}
\usepackage[pdftex]{graphicx}
\usepackage{textcomp}
\usepackage{color}
\usepackage{xcolor}
\usepackage{multirow}
\usepackage{multicol,lipsum}
\usepackage[utf8x]{inputenc}
\usepackage{epsfig}
\usepackage{algorithm}
\usepackage{algorithmic}
\usepackage{multirow}
\usepackage{amsmath}
\usepackage{hyperref}
\usepackage{mathtools}

\usepackage{bm}
\usepackage{listings}
\newtheorem{remark}{Remark}
\newtheorem{definition}{Definition}
\newtheorem{theorem}{Theorem}

\newtheorem{problem}{Problem}
\newtheorem{example}{Example}

\newtheorem{assumption}{Assumption}
\usepackage{balance}
\usepackage{textcomp}
\usepackage{xcolor}
\def\BibTeX{{\rm B\kern-.05em{\sc i\kern-.025em b}\kern-.08em
    T\kern-.1667em\lower.7ex\hbox{E}\kern-.125emX}}
    
\newcommand{\oomit}[1]{}
\ifCLASSOPTIONcompsoc
  \usepackage[nocompress]{cite}
\else
  \usepackage{cite}
\fi
\ifCLASSINFOpdf
\else
\fi
\begin{document}

%
% paper title
% Titles are generally capitalized except for words such as a, an, and, as,
% at, but, by, for, in, nor, of, on, or, the, to and up, which are usually
% not capitalized unless they are the first or last word of the title.
% Linebreaks \\ can be used within to get better formatting as desired.
% Do not put math or special symbols in the title.
\title{PAC Model Checking of Black-Box Continuous-Time Dynamical Systems}
%
%
% author names and IEEE memberships
% note positions of commas and nonbreaking spaces ( ~ ) LaTeX will not break
% a structure at a ~ so this keeps an author's name from being broken across
% two lines.
% use \thanks{} to gain access to the first footnote area
% a separate \thanks must be used for each paragraph as LaTeX2e's \thanks
% was not built to handle multiple paragraphs
%
%
%\IEEEcompsocitemizethanks is a special \thanks that produces the bulleted
% lists the Computer Society journals use for "first footnote" author
% affiliations. Use \IEEEcompsocthanksitem which works much like \item
% for each affiliation group. When not in compsoc mode,
% \IEEEcompsocitemizethanks becomes like \thanks and
% \IEEEcompsocthanksitem becomes a line break with idention. This
% facilitates dual compilation, although admittedly the differences in the
% desired content of \author between the different types of papers makes a
% one-size-fits-all approach a daunting prospect. For instance, compsoc 
% journal papers have the author affiliations above the "Manuscript
% received ..."  text while in non-compsoc journals this is reversed. Sigh.

\author{Bai Xue,~\IEEEmembership{Member,~IEEE}, Miaomiao Zhang, Arvind  Easwaran and Qin Li
\thanks{Corresponding Authors: Bai Xue and Miaomiao Zhang}
\thanks{B. Xue is with State Key Lab. of Computer Science, Institute of Software, CAS, and University of Chinese Academy of Sciences, Beijing, China email: xuebai@ios.ac.cn}
\thanks{Miaomiao Zhang is with School of Software Engineering, Tongji University, China email: miaomiao@tongji.edu.cn}
\thanks{A. Easwaran is with School of Computer Science and Engineering, Nanyang Technological University (NTU), Singapore email: arvinde@ntu.edu.sg.}
\thanks{Q. Li is with Shanghai Key Laboratory of Trustworthy Computing
East China Normal University, Shanghai, China email: qli@sei.ecnu.edu.cn.}
\thanks{
Manuscript received April 17, 2020; revised June 17, 2020; accepted July 6, 2020. This article was presented in the International Conference on Embedded Software 2020 and appears as part of the ESWEEK-TCAD special issue.}
\thanks{This work has been supported through grants by NSFC under grant No. 61872341, 61836005, 61972284, the CAS Pioneer Hundred Talents Program under grant No. Y8YC235015, the MoE, Singapore, Tier-2 grant \#MOE2019-T2-2-040, and the foundation of Shenzhen Institute of Artificial Intelligence and Robotics for Society and the foundation of National Trusted Embedded Software Engineering Technology Research Center.}
}

\maketitle

\begin{abstract}
In this paper we present a novel model checking approach to finite-time safety verification of black-box continuous-time dynamical systems within the framework of probably approximately correct (PAC) learning. The black-box dynamical systems are the ones, for which no model is given but whose states changing continuously through time within a finite time interval can be observed at some discrete time instants for a given input. The new model checking approach is termed as PAC model checking due to incorporation of learned models with correctness guarantees expressed using the terms error probability and confidence. Based on the error probability and confidence level, our approach provides statistically formal guarantees that the time-evolving trajectories of the black-box dynamical system over finite time horizons  fall within the range of the learned model plus a bounded interval, contributing to insights on the reachability of the black-box system and thus on the satisfiability of its safety requirements. The learned model together with the bounded interval is obtained by scenario optimization, which boils down to a linear programming problem. Three examples demonstrate the performance of our approach.
\end{abstract}

% To allow for easy dual compilation without having to reenter the
% abstract/keywords data, the \IEEEtitleabstractindextext text will
% not be used in maketitle, but will appear (i.e., to be "transported")
% here as \IEEEdisplaynontitleabstractindextext when compsoc mode
% is not selected <OR> if conference mode is selected - because compsoc
% conference papers position the abstract like regular (non-compsoc)
% papers do!
\IEEEdisplaynontitleabstractindextext
% \IEEEdisplaynontitleabstractindextext has no effect when using
% compsoc under a non-conference mode.

% For peer review papers, you can put extra information on the cover
% page as needed:
% \ifCLASSOPTIONpeerreview
% \begin{center} \bfseries EDICS Category: 3-BBND \end{center}
% \fi
%
% For peerreview papers, this IEEEtran command inserts a page break and
% creates the second title. It will be ignored for other modes.
\IEEEpeerreviewmaketitle

\section{Introduction}
\label{Int}
The complexity of today's technological applications induces a quest for automation, leading to many black-box intelligent cyber-physical systems and thus being difficult to reason about \cite{lee2008cyber}. Many of these systems operate in safety-critical context and hence safety-critical systems themselves \cite{rajkumar2010cyber}. Therefore, reasonable performance guarantees should be obtained before the systems are deployed.

Black-box checking, introduced by Peled at al. \cite{peled1999black}, is often used for verifying non-stochastic black-box systems, based on experiments that interface with them. It performs checks on the system itself. The black-box checking is a combination of model checking and testing: model checking \cite{clarke1994model} checks properties of a model of the system, but not the system itself. In contrary, testing is usually applied to the actual system and checks whether the system conforms with the model, further serving to improve the model. They are two complementary approaches for enhancing the reliability of black-box systems. In the black-box checking, whenever a model is created, model checking may reveal a fault in the system or show that the model was not good enough and needs to be learned further if the fault is spurious. If model checking does not reveal a fault, equivalence between the model and the black-box system is checked via testing. In case, non-equivalence is detected, then the model needs to be further learned. The checking-testing-learning repeated process is costly generally. Recently, a method combining optimization-based falsification and black-box
checking was proposed to falsify specifications for black-box cyber-physical systems in \cite{waga2020}.

Another technique to verification of black-box systems is statistical model checking (SMC) \cite{sen2004statistical,younes2005}. SMC is pioneered by Younes and
Simmons in the discrete case in \cite{younes2002}, which is based on Sequential Probability Ratio Test \cite{Wald1945}. It is a compromise between verification and testing, which is based on sampling executions of the system and then deciding whether the samples provide a statistical evidence for the satisfaction or violation of the specification based on hypothesis testing \cite{ReijsbergenBSH15}. SMC is now widely accepted in various research areas such as software engineering, in particular for industrial applications \cite{clarke2011}, or even for solving problems originating from systems biology \cite{clarke2008}. There are several reasons for this success. First, SMC is very simple to understand, implement and use. Second, it does not require extra modelling or specification effort, but simply an executable system that can be simulated and checked against state-based properties. Third, it avoids the state space explosion in verification and thus can be applied to analyze systems with large state spaces. Consequently, there are variety of SMC tools such as PLASMA-Lab \cite{boyer2013plasma}, Ymer \cite{younes2005ymer}, VeStA \cite{SenVA05}, MRMC \cite{katoen2011ins}, MC2 \cite{grosu2005monte}, UPPAAL-SMC \cite{david2015uppaal} and so on. In order to further improve the efficiency of SMC, Bayesian SMC was proposed in \cite{jha2009bayesian,zuliani2013}, which is a SMC based on Bayesian statistics.  The aforementioned SMC approaches for black-box systems are free of mathematical models and perform checks on the system itself by sampling executions of the system. However, the usefulness of mathematical models is well documented. The mathematical models not only help us to understand the system, but also are instrumental to yield insight into the complex processes involved in the system by extracting the essential meaning of some hypotheses. Also, they allow to study the effects of changes in their components and/or environmental conditions on the system's trajectories, i.e., they allow the control and optimization of the system. Thus, the introduction of mathematical models with appropriate degree of complexity into SMC would contribute a lot to the analysis of the black-box system, not only in the verification of its specifications but also in understanding the complex mechanisms underlying and thus further optimizing the system. \oomit{By extending the checking-testing-learning repeated process in black-box checking to probabilistic systems,}Consequently, model learning based SMC approaches are also proposed. For example, \cite{mao2011learning,mao2012learning,mao2016learning,aichernig2019} considered black-box systems modelled by Markov decision processes and inferred probabilistic models with the purpose of model checking. The work in \cite{nouri2014faster} combined stochastic learning and abstraction with respect to some property for analyzing black-box systems modelled by Markov decision processes. The work in \cite{brazdil2014} presented an approach for black-box systems modelled by Markov decision processes to unbounded reachability analysis via SMC. The technique is based on delayed Q-learning, a form of reinforcement learning. Generally, the exact learning algorithms require
checking equivalence between the model and the system, which is difficult and undecidable. Regression models were used in \cite{ fan2020parameter} for finding the regions in the parameter space that lead to satisfaction or violation of given specification with probabilistic coverage guarantees based on conformal regression. Recently, learning procedure within the PAC learning framework is proposed, e.g.,  \cite{fu2014probably,chen2016pac,ashok2019pac,park2019pac}.

In this paper we propose a novel SMC approach for finite-time safety verification of black-box continuous-time dynamical systems within the framework of PAC learning \cite{franzle2015multi}. The black-box continuous-time dynamical systems are the ones, for which no model is given but whose states changing continuously through time over finite time horizons can be observed at some discrete time instants for a given input. The proposed new model checking, also termed as PAC model checking, is built upon learned models within the framework of PAC learning. In the PAC model checking, correctness guarantees of the learned models are expressed using the terms error probability and confidence level. We show that the time-evolving trajectories of the black-box system over a specified finite time horizon fall within the range of the learned model plus a bounded interval with statistical guarantees, which is further used to characterize the satisfiability of safety requirements. Given an error probability and a confidence level, which are two fundamental parameters in PAC learning, the model together with the bounded interval is computed via scenario optimization, which is widely used for computing solutions to robust optimization problems based on finite randomization of infinite constraints \cite{calafiore2006}. The scenario optimization, which finally boils down to a linear program in our approach, is constructed from a family of independent and identically distributed datum collected by executing the system. Three examples demonstrate the performance of our approach. Our contributions are summarized as follows.

1). We propose a novel PAC model checking approach for finite-time safety verification of black-box continuous-time dynamical systems. In this approach the trajectories of the black-box system over finite time horizons are shown to fall within the range of a model plus a bounded interval with error probabilities and confidence levels. This reachability analysis is instrumental in characterizing the satisfiability of safety requirements of the black-box system.

2). A linear programming based approach is proposed to synthesize the model and the bounded interval. The size of the linear programming problem could be independent of the one of the black-box system, thus rendering our approach suitable for large-scale systems.  

\subsection*{Related Work}
As mentioned above, there are many works on verifying black-box systems. In this subsection we just discuss the closely related works to the present one.

The works \cite{fu2014probably,ashok2019pac} considered (unbounded) reachability for Markov decision processes (and stochastic games in \cite{ashok2019pac}) and inferred the transition probabilities with PAC guarantees. The work \cite{park2019pac} proposed an algorithm for constructing PAC confidence sets for deep neural networks. The work in \cite{xue2019safe} computed safe inputs for a black-box system such that the system's final outputs fall within a safe range with PAC guarantees. In contrast, our approach focuses on analysis of continuous-time systems, and infers that the time-evolving trajectories of the black-box system over finite time horizons fall within the range of a model plus a bounded interval with PAC guarantees. The closest work in spirit to the present one is  \cite{chen2016pac}, which considered verification of sequential programs by learning models of the set of feasible paths of programs within the framework of PAC learning. The model learning algorithm in \cite{chen2016pac} is based on counterexample guided abstraction refinement. However, our approach considers continuous-time systems and infers an approximation to the trajectories of the system over the specified finite time horizon within the framework of PAC learning, in which linear programs are used for learning models. 

In the framework of simulation-driven reachability analysis \cite{DuggiralaMVP15}, a PAC based method was proposed for learning discrepancy functions in \cite{FanQM017} for safety verification of hybrid systems with black-box modules. The problem of learning discrepancy functions is reduced to a problem of learning linear separators. Although a PAC discrepancy function is computed in \cite{FanQM017}, a characterization on how well the trajectories satisfy the learned discrepancy function is not given and thus a formal quantitative assessment on the satisfiability of safety properties is not presented if a valid discrepancy function is not obtained. Generally, valid discrepancy functions rather than PAC ones for black-box systems are challenging to obtain. In contrast, a formal characterization of the satisfiability of safety properties is given based on the computation of PAC models in our PAC model checking method.

When the continuous-time systems of interest are modeled by ordinary differential equations or delay differential equations, and the equations are explicitly given, there are many well-developed model-based reachability analysis techniques over finite time horizons, e.g., Taylor-model method \cite{ChenAS13}, simulation-driven reachability method \cite{DuggiralaMVP15} and set-boundary reachability method \cite{9023360}, for safety verification of these systems. However, our method focuses on black-box continuous-time dynamical systems, whose mathematical abstractions are not acquired and which are only represented by a family of datum. Such systems can not be handled by  existing model-based reachability analysis techniques.

\medskip
The remainder of this paper is structured as follows. In Section \ref{Pre} we formalize the concept of black-box continuous-time dynamical systems and the problem of interest in this paper. Section \ref{PBEG} elucidates our PAC model checking approach. After demonstrating the performance of our approach on three examples in Section \ref{experiments}, we conclude this paper in Section \ref{conclusion}.

\section{Preliminaries}
\label{Pre}
In this section we present the concept of black-box continuous-time dynamical systems and the related problems, as well as a brief introduction on scenario optimization. The notations are used throughout this paper: $\mathbb{R}_{\geq 0}$ denotes the set of nonnegative real values. $\mathbb{R}_{>0}$ denotes the set of positive real values. Vectors are denoted by boldface letters. \textit{Besides, the ground truth trajectories in all examples are obtained based on the combination of Runge-Kutta simulation methods and linear interpolation methods. }
\subsection{Problem Formulation}
\label{PF}
In this paper we consider a black-box continuous-time dynamical system, whose dynamics are governed by a formula of the following form:
\begin{equation}
\label{bb}
y(t)=b(\bm{x}_0,t), 
\end{equation}
where $\bm{x}_0=(x_{0,1},\ldots,x_{0,n})^{\top}\in \mathcal{X}_0$ is the input of the system, the set $\mathcal{X}_0\subseteq \mathbb{R}^n$ is compact, $t\in [0,T]$ with $T\in \mathbb{R}_{>0}$ is the time variable, $y(t)$ is the state of the system at time $t$, and $b(\cdot,\cdot): \mathcal{X}_0\times [0,T]\rightarrow \mathbb{R}$ is the system mapping which is unknown. Besides, we have the following assumptions.

\begin{assumption}
1). The system \eqref{bb} runs well, including the on-board sensors, and thus it can provide us any family of finite datum we need. Also, the provided datum are free of noise.

2). Suppose that the time horizon $[0,T]$ is endowed with a $\sigma-$algebra $\mathcal{D}_t$ and a probability $P_t$ over $\mathcal{D}_t$ is assigned.  Also, we assume that the set $\mathcal{X}_0$ of inputs is endowed with a $\sigma-$algebra $\mathcal{D}_{\bm{x}_0}$ and that a probability $P_{\bm{x}_0}$ over $\mathcal{D}_{\bm{x}_0}$ is assigned. Throughout this paper, we use the uniform distribution $P_t$ on $[0,T]$ and $P_{\bm{x}_0}$ on $\mathcal{X}_0$ to illustrate our method, although our method is not confined to this particular distribution. 
\end{assumption}
  
The system \eqref{bb} is illustrated in Fig. \ref{illustration1}.  Given an input $\bm{x}_0 \in \mathcal{X}_0$, the trajectory of the system \eqref{bb} with the input $\bm{x}_0$ is denoted by $y_{\bm{x}_0}(\cdot): [0,T] \rightarrow \mathbb{R}$. 

\begin{figure}
\centering
\includegraphics[width=0.4\textwidth]{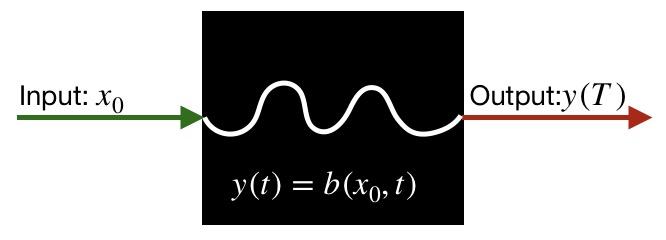} 
\caption{An illustration of the system \eqref{bb}.}
\vspace{-0.4cm}
\label{illustration1}
\end{figure}

Systems of the form \eqref{bb} are all around us, especially nowadays. For example, many AI systems such as robotics and self-driving cars are leaving academic laboratories and entering real-world applications. Unfortunately, many of these systems can not explain their results even to their makers, let alone to end-users \cite{castelvecchi2016}. They operate like black boxes, which can be viewed in terms of a family of observed datum, without any knowledge of their internal workings.

In this paper we propose a PAC model checking approach for finite-time  safety verification of the system \eqref{bb}. The safety verification problem is widely studied in computer science, e.g.,\cite{immler2015}. In our approach, the key is to obtain a model with appropriate degree of complexity, which is learned based on a family of collected datum within the framework of PAC learning and can characterize the system \eqref{bb} with correctness guarantees expressed with error probabilities and confidence levels. For computing such models, we should address the problems summarized below:
\begin{problem}
\label{prob}
\begin{enumerate}
\item[1.1] What datum should we use?
\item[1.2] How can we learn a mathematical model efficiently based on the collected datum?
\item[1.3] What is the discrepancy between the trajectories  of the learned mathematical model and the system \eqref{bb}?
\end{enumerate}
\end{problem}
After computing the model, we will address the safety verification problem below. 
\begin{problem}
\label{veri}
Given a set $\mathtt{Uns}\subseteq \mathbb{R}$ of unsafe states,  when the trajectories of the computed model are shown to avoid the set $\mathtt{Uns}$, how can we formally characterize the satisfiability of the safety property of avoiding the unsafe set $\mathtt{Uns}$ for the black-box system \eqref{bb} over the time horizon $[0,T]$?
\end{problem}

We in the sequel solve Problems \ref{prob} and \ref{veri} based on scenario optimization.

\begin{remark}
Our method can be  straightforwardly extended to vector valued mappings of the form $\bm{b}(\cdot,\cdot):\mathcal{X}_0\times [0,T] \rightarrow \mathbb{R}^{q}$ with $q >1$, but the scalar valued mappings $b(\cdot,\cdot):\mathcal{X}_0\times [0,T] \rightarrow \mathbb{R}$ are considered for ease of exposition.
\end{remark}

\subsection{Scenario Optimization}
\label{SO}
This subsection gives a brief introduction on scenario optimization. It provides statistical solutions to  robust optimization problems based on solving finite randomization of infinite convex constraints. 

A robust optimization problem of interest is as follows:
\begin{equation}
\label{rop}
\begin{split}
&\min_{\bm{\gamma}\in  \Gamma \subseteq \mathbb{R}^m} \bm{c}^{\top} \bm{\gamma}\\
&\text{\rm s.~t.~}\bm{f}_{\bm{\delta}}(\bm{\gamma})\leq 0, \forall \bm{\delta} \in \Delta,
\end{split}
\end{equation}
where $\bm{f}_{\bm{\delta}}(\bm{\gamma})$ are continuous and convex functions over the $m-$dimensional optimization variable $\bm{\gamma}$ for every $\bm{\delta} \in \Delta$. Also, the sets $\Gamma$ and $\Delta$  are convex and closed.

Generally, it is challenging to solve \eqref{rop}. The work in \cite{calafiore2006} proposed a scenario optimization approach for solving \eqref{rop} with statistically formal guarantees.
\begin{definition}
\label{sop}
Suppose that  $\Delta$ is endowed with a $\sigma-$algebra $\mathcal{D}$ and that a probability {\rm P} over $\mathcal{D}$ is assigned. The scenario optimization of \eqref{rop} is to obtain an approximate solution to \eqref{rop} via solving the convex program \eqref{cop}, which is constructed by extracting $K$ independent and identically distributed samples $(\bm{\delta}_i)_{i=1}^K$ from $\Delta$ according to the probability distribution {\rm P}:
\begin{equation}
\label{cop}
\begin{split}
&\min_{\bm{\gamma} \in \Gamma \subseteq \mathbb{R}^m} \bm{c}^{\top} \bm{\gamma}\\
&\text{\rm s.~t.~}\wedge_{i=1}^K\bm{f}_{\bm{\delta}_{i}}(\bm{\gamma}) \leq 0.
\end{split}
\end{equation}
\end{definition}

 \eqref{cop} relaxes \eqref{rop} in that it only considers a finite subset of the infinitely many constraints of \eqref{rop}. A mathematically rigorous relation, which holds irrespective of the underlying probability P, between the solutions of the two systems can be drawn  \cite{campi2009}. 
\begin{theorem} 
\label{sec1}
 If \eqref{cop} is feasible and attains a unique optimal solution $\bm{\gamma}_K^*$, and 
\begin{equation}
\label{N}
\epsilon\geq \frac{2}{K}(\ln \frac{1}{\beta}+m),
\end{equation}
where $\epsilon\in (0,1)$ and $\beta \in (0,1)$ are respectively a user-chosen error level and confidence level, then with at least $1-\beta$ confidence, $\bm{\gamma}_K^*$ satisfies all constraints in $\Delta$ but at most a fraction of probability measure $\epsilon$, i.e., $\text{\rm P}(\{\bm{\delta}\in \Delta\mid \bm{f}_{\bm{\delta}}(\bm{\gamma}_K^*)\nleq 0\}) \leq \epsilon$, where the confidence $\beta$ is the $K-$fold probability $\text{\rm P}^K$ in $\Delta^K=\Delta\times \ldots \times \Delta$, which is the set to which the extracted sample $(\bm{\delta}_1,\ldots,\bm{\delta}_K)$ belongs. 
\end{theorem}

The above conclusion still holds if the uniqueness of optimal solutions to \eqref{cop} is removed \cite{calafiore2006}, since a unique optimal solution can always be obtained according to Tie-break rule if multiple optimal solutions occur. Moreover, since $\beta$ appears under the sign of logarithm in \eqref{N}, it can be made small, like $10^{-10}$ or $10^{-20}$, without increasing $K$ significantly. Recently, scenario optimization was used to compute probably approximately safe inputs for a black-box system such that the system's final outputs fall within a safe range in \cite{xue2019safe}, and perform safety verification of hybrid systems in \cite{xue2019probably}. 
 \section{PAC Model Checking}
\label{PBEG}
In this paper we present our PAC model checking approach for safety verification of the black-box system \eqref{bb} by solving Problems \ref{prob} and \ref{veri}. 

\subsection{Datum Extraction}
\label{ED}
In this subsection we introduce what datum to use in learning a model of the system \eqref{bb} in our approach and how to obtain them, i.e., solve Problem 1.1.

We first extract a family of independent and identically distributed time instances $(t_j)_{j=1}^M$ from the time interval $[0,T]$ according to the probability distribution $P_t$. Moreover, a family of independent and identically distributed inputs $(\bm{x}_{0,i})_{i=1}^N$ is also extracted from the set $\mathcal{X}_0$ according to the probability distribution $P_{\bm{x}_0}$. The process of obtaining $(t_j)_{j=1}^M$ and $(\bm{x}_{0,i})_{i=1}^N$ does not need to run or /simulate the system \eqref{bb}. The numbers $M$ and $N$ rely on how accurate one wants the learned model to achieve. The relationship is elucidated in Subsection \ref{FELM}.

Next we need to run the system \eqref{bb} to obtain its internal datum.  For each extracted input $\bm{x}_{0,i}$, $i=1,\ldots, N$, we feed it to the system \eqref{bb} and then run it until the time $T$. In this process, the on-board sensors will help observe and record the states of the system \eqref{bb} at the time instance $t_j$, $j=1,\ldots,M$.  This is realistic for some systems nowadays, since smart sensors are taking over almost every sphere of human life.  For example, RADAR, LIDAR, GPS and computer vision are widely used to work coherently for identifying the position, velocity and other states of the vehicle. We denote the family of observed states by $(y_{i,j})_{i=1,\ldots,N,j=1,\ldots,M}$, where $y_{i,j}$ denotes the state of the system \eqref{bb} at time $t_j$  with the input $\bm{x}_{0,i}$, $i=1,\ldots,N$, $j=1,\ldots,M.$

So far, we obtain a family of datum $\Big((\bm{x}_{0,i},t_j,y_{i,j})\Big)_{j=1,\ldots,M}^{ i=1,\ldots,N,}$. Each data is a triple $(\bm{x}_{0},t,y(t))$, where $\bm{x}_0$ is the input of the system \eqref{bb},  $t\in [0,T]$ is the time instance and $y(t)$ is the state of the system \eqref{bb} with the input $\bm{x}_0$ at time $t$. The process of running the system \eqref{bb} can be regarded as a testing process. However, our method goes further than testing techniques. We meanwhile collect a family of datum and then use these datum to compute models for characterizing the system \eqref{bb} formally.

In our experiment, we assume that the input $\bm{x}_{0,i}$ is noise-free and the on-board sensors work perfectly such that the observed datum are free of noise as well, i.e., $y_{i,j}$ is the exact state of the system \eqref{bb} with the input $\bm{x}_{0,i}$ at time $t=t_{j}$, $i=1,\ldots,N$, $j=1,\ldots,M.$ This assumption may be too ideal in practice since input and sensor noise often exists. We would relax it in our future work.

\subsection{Safety Verification}
\label{FELM}
In this section we elucidate our approach for solving Problems 1.2, 1.3 and 2 based on the family of datum obtained from the process in Subsection \ref{ED}. We first consider the system \eqref{bb} with one trajectory, and then multiple trajectories and finally all trajectories from the input set $\mathcal{X}_0$.  

\subsubsection{{\rm One Trajectory Verification}}
\label{OTC}
In this subsection, we solve Problems 1.2, 1.3 and 2 for the system \eqref{bb} with a single input. Concretely, given a discrete-time trajectory of the system \eqref{bb} with the input $\bm{x}_{0,i}$, which is represented by a family of datum $\Big((\bm{x}_{0,i},t_j,y_{i,j})\Big)_{j=1}^M$ with $(t_j)_{j=1}^M$ and $(y_{i,j})_{j=1}^M$ obtained in Subsection \ref{ED}, we would compute a model  $z(t)=w(\bm{x}_{0,i},t)$ with $w(\bm{x}_{0,i},\cdot):[0,T] \rightarrow \mathbb{R}$ to characterize $y_{\bm{x}_{0,i}}(\cdot):[0,T]\rightarrow \mathbb{R}$.

\subsubsection*{{\rm PAC Models}}
\label{pacm}
In computing a model, we consider a linearly-parameterized model template  $w(c_{1},\ldots, c_{k}, \bm{x}_{0,i},t)$, $k\geq 1$ such that $w(c_{1},\ldots, c_{k},\bm{x}_{0,i},t)$ is for $t\in [0,T]$ a linear  function in $c_{1},\ldots, c_{k}$, which are unknown parameters. This model can be a polynomial function over $t$, or a more general nonlinear function over $t$. For instance, consider a two-dimensional system with input state variable $\bm{x}=(x_1,x_2)^{\top}$, $w(c_{1},c_{2}, \bm{x},t)=c_1x_1t+c_2x_2t^2$ is a linear  function in $c_1$ and $c_2$, and $w(c_{1},c_{2}, \bm{x},t)=c_1e^{x_1 x_2}t+c_2\ln{(x_2t^2)}$ is also a linear function over $c_1$ and $c_2$. Such models can be the ones parameterized with orthonormal basis functions, which are able to represent a set of physical systems \cite{horst2013global}. For ease of exposition, we use $\bm{c}$ to denote $(c_{l})_{l=1,\ldots,k}$ in the reminder of this paper. \textit{Generally, a model template of appropriate degree of complexity should be chosen in order to avoid the over-fitting issue and facilitate the reachability analysis.} In practice, engineering insight and physical knowledge would facilitate the selection of model templates.

Then we construct the following linear program over $\bm{c}$ for computing a mathematical model based on the family of  given datum $\Big((\bm{x}_{0,i},t_j,y_{i,j})\Big)_{j=1}^M$:
\begin{equation}
\begin{split}
&\min_{\bm{c},\xi} \xi\\
&\text{s.~t.~for each~} j=1,\ldots, M:\\
& w(\bm{c}, \bm{x}_{0,i},t_j)-b(\bm{x}_{0,i},t_j)\leq \xi,\\
& b(\bm{x}_{0,i},t_j)-w(\bm{c}, \bm{x}_{0,i},t_j)\leq \xi,\\
&-U_c\leq c_{l}\leq U_c, l=1,\ldots,k,\\
& 0\leq \xi \leq U_{\xi},
\end{split}
\end{equation}
which is equivalent to 
 \begin{equation}
 \label{lp}
\begin{split}
&\min_{\bm{c},\xi} \xi\\
&\text{s.~t.~for each~} j=1,\ldots, M:\\
& w(\bm{c}, \bm{x}_{0,i},t_j)-y_{i,j}\leq \xi,\\
& y_{i,j}-w(\bm{c}, \bm{x}_{0,i},t_j)\leq \xi,\\
&-U_c\leq c_{l}\leq U_c, l=1,\ldots,k,\\
& 0\leq \xi \leq U_{\xi},
\end{split}
\end{equation}
where $U_c\in \mathbb{R}_{\geq 0}$ is a pre-specified upper bound for $c_{l}$, $l=1,\ldots,k$, and $U_{\xi}\in \mathbb{R}_{\geq 0}$ is a pre-specified upper bound for $\xi$.

Denote the optimal solution to \eqref{lp} by $(\bm{c}^*,\xi^{*})$. Thus, we obtain a model $z(t)=w(\bm{c}^*,\bm{x}_{0,i},t)$, whose discrepancy with the system \eqref{bb} is characterized by two approximation parameters: error probability  $\epsilon \in (0,1)$ and confidence level $\beta \in (0,1)$.   This is formally stated in Theorem \ref{conclusion1}.
\begin{theorem}
\label{conclusion1}
Let $(\bm{c}^*,\xi^*)$ be an optimal solution to \eqref{lp}, $\epsilon\in (0,1)$, $\beta\in (0,1)$ and 
\begin{equation}
\label{NLP}
\epsilon\geq \frac{2}{M}(\ln \frac{1}{\beta}+k+1). 
\end{equation} 
Then we have that with at least $1-\beta$ confidence,  
\begin{equation}
\label{pro}
P_t\Big(\left\{t\in [0,T]\middle|\;
\begin{aligned}
&|w(\bm{c}^*, \bm{x}_{0,i},t)-b(\bm{x}_{0,i},t)|\\
&~~~~~~~~~~~~~~~~~~~~~~~~~~~~~~~~\leq \xi^*
\end{aligned}
\right\}\Big)\geq 1-\epsilon.
\end{equation}
\end{theorem}
\begin{proof}
The conclusion is easily obtained by Theorem \ref{sec1}.
\end{proof}

Actually, the computed mathematical model $z(t)=w(\bm{c}^*,\bm{x}_{0,i},t)$ is a PAC model \cite{valiant2013,shalev2014} with accuracy level $\epsilon$ and confidence level $\beta$. The accuracy parameter $\epsilon$ in Theorem \ref{conclusion1} determines how far the learned model can be from the real one. This corresponds to the "approximately correct". A confidence parameter $\beta$ indicates how likely the learned model is to meet that accuracy requirement. This corresponds to the "probably" part. Under the data access model that we are investigating, these approximations are inevitable. Since the training set $\Big((\bm{x}_{0,i},t_j,y_{i,j})\Big)_{j=1}^M$ is randomly generated, there may always be a small chance that it will happen to be noninformative (for example, there is always some chance that the training set will contain only one domain point, sampled over and over again). Furthermore, even when we are lucky enough to get a training sample that does faithfully represent $[0,T]$, because it is just a finite sample, there may always be some finite details of $[0,T]$ that it fails to reflect. The accuracy parameter $\epsilon$ allows forgiving the learned model for making minor errors. 

\subsubsection*{{\rm \textbf{One Trajectory Verification}}}
\label{PACCOT}
Based on Theorem \ref{conclusion1}, we in this subsection solve Problem \ref{veri} for the system \eqref{bb} with one trajectory $y_{\bm{x}_0,i}(\cdot):[0,T]\rightarrow \mathbb{R}$ using the trajectory of the mathematical model $z(t)=w(\bm{c}^*,\bm{x}_{0,i},t)$ within the framework of PAC learning.

We first characterize the reachability of the trajectory $y_{\bm{x}_0,i}(\cdot):[0,T]\rightarrow \mathbb{R}$ using the mathematical model $z(t)=w(\bm{c}^*,\bm{x}_{0,i},t)$ plus the computed $\xi^*$. We denote the trajectory of the mathematical model $z(t)=w(\bm{c}^*,\bm{x}_{0,i},t)$ by $z_{\bm{x}_{0,i}}(\cdot): [0,T]\rightarrow \mathbb{R}$.  From Theorem \ref{conclusion1}, we have that with confidence of at least $1-\beta$, 
\begin{equation}
\label{pro1}
y_{\bm{x}_{0,i}}(t)\in [z_{\bm{x}_{0,i}}(t)-\xi^*, z_{\bm{x}_{0,i}}(t)+\xi^*]
\end{equation}
 for all $t$ in  $[0,T]$ but at most a fraction of probability measure $\epsilon$, i.e., with confidence of at least $1-\beta$, the amount of time for the trajectory $y_{\bm{x}_{0,i}}(\cdot):[0,T]\rightarrow \mathbb{R}$ staying within the $\xi^{*}$ neighborhood of the trajectory $z_{\bm{x}_{0,i}}(\cdot):[0,T] \rightarrow \mathbb{R}$ exceeds $T(1-\epsilon)$. A graph explanation is further presented in Fig. \ref{illustration2} to enhance the understanding of \eqref{pro1}. In Fig. \ref{illustration2}, $y_{\bm{x}_{0,i}}(t)\notin [z_{\bm{x}_{0,i}}(t)-\xi^*, z_{\bm{x}_{0,i}}(t)+\xi^*]$ for $t \in [t_1,t_2]\cup[t_3,t_4] \cup [t_5,t_6]$. According to Theorem \ref{conclusion1}, $t_{6}-t_5+t_{4}-t_3+t_{2}-t_1 \leq \epsilon T$ with confidence of at least $1-\beta$.

\begin{figure}
\centering
\includegraphics[width=3.05in,height=1.3in]{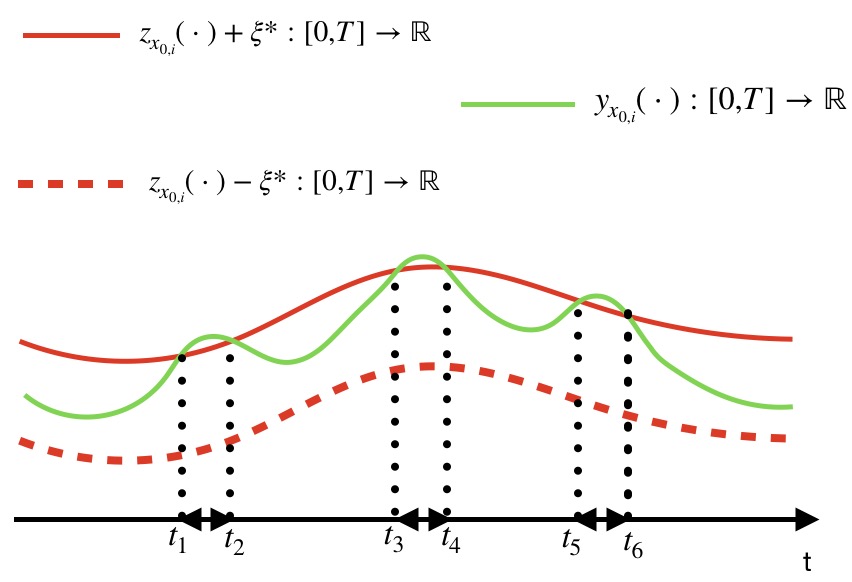} 
\caption{An illustration of the discrepancy between the mathematical model $z(t)=w(\bm{c}^*,\bm{x}_{0,i},t)$ and the system $y(t)=b(\bm{x}_{0,i},t)$ for $t\in [0,T]$.}
\label{illustration2}
\end{figure}

Then we solve Problem \ref{veri} based on the formal reachability characterization given above. That is, 

\textit{
if $[z_{\bm{x}_{0,i}}(t)-\xi^*, z_{\bm{x}_{0,i}}(t)+\xi^*]$ does not intersect the unsafe set $\mathtt{Uns}$ for $t\in [0,T]$, i. e.,
 $[z_{\bm{x}_{0,i}}(t)-\xi^*, z_{\bm{x}_{0,i}}(t)+\xi^*]\cap \mathtt{Uns}=\emptyset$ for $t\in [0,T]$, we have that the amount of time the system \eqref{bb} with the input $\bm{x}_{0,i}$ spends inside the unsafe set $\mathtt{Uns}$  does not exceed $\epsilon T$, with confidence of at least $1-\beta$.  
 }
 
 If $\beta$ in Theorem \ref{conclusion1} is extremely small (smaller than $10^{-20}$), then we have a priori practical certainty that the total amount of unsafe time does not exceed $\epsilon T$.  As explained in Subsection \ref{SO}, the confidence level $1-\beta$ can be made large without increasing the size $M$ of samples significantly. This framework is useful in those situations where the system \eqref{bb} is able to tolerate the exposure to a deteriorating agent for a limited amount of time.  For example, let us consider a solar-powered autonomous vehicle. Regions without solar exposure are considered to be unsafe, since the vehicle's battery could be drained after a period of time. However, it would be inefficient to plan a path for the vehicle completely avoiding  all these shaded regions. Instead, a more reasonable requirement would be that the amount of time the vehicle spends in the shaded regions is small.   
 
 \begin{remark}
 Our approach can also be used to characterize the case that there exists $t\in [0,T]$ such that $[z_{\bm{x}_{0,i}}(t)-\xi^*, z_{\bm{x}_{0,i}}(t)+\xi^*]\cap \mathtt{Uns}\neq \emptyset$. For this case, we need to compute a value $\tau\geq 0$, which is larger than or equal to the amount of time such that $[z_{\bm{x}_{0,i}}(t)-\xi^*, z_{\bm{x}_{0,i}}(t)+\xi^*]\cap \mathtt{Uns}\neq \emptyset$. Further, we have that the amount of time the system \eqref{bb} with the input $\bm{x}_{0,i}$ spends inside the unsafe set $\mathtt{Uns}$ does not exceed $\epsilon T+\tau$, with confidence of at least $1-\beta$. 
 \end{remark}

In the following we use an example from a Van-der-Pol oscillator to enhance the understanding of our approach.
\begin{example}
\label{van}
Consider a system with $T=10$, $\bm{x}_{0,i}=(1.4,2.3)^{\top}$ and $\mathtt{Uns}=\{y\in \mathbb{R}\mid y\geq 3\}$, whose internal dynamics are described by an ordinary differential equation which generally describes a Van-der-Pol oscillator \cite{van1926}:
\begin{equation}
\label{ode1}
    \begin{cases}
    \frac{d x_1}{d t}=x_2\\
    \frac{d x_2}{d t}=(1-x_1^2)x_2-x_1
    \end{cases}.
\end{equation}

We assume that the trajectory of the system \eqref{bb} in this example describes the time evolution of the state $x_1$ in \eqref{ode1}, i.e., $y(t)=b(\bm{x}_{0,i},t)=x_1(t)$ for $t\in [0,10]$. The ground truth trajectory $y_{\bm{x}_{0,i}}(\cdot):[0,T]\rightarrow \mathbb{R}$, is illustrated in Fig. \ref{fig}.  It is used to extract datum $\Big((\bm{x}_{0,i},t_j,y_{i,j})\Big)_{j=1}^M$ and perform comparisons.  The method of constructing the ground truth trajectory is introduced in the beginning of Section \ref{Pre}. 

Let $\beta=10^{-20}$ and $\epsilon=0.01$. In this example we use $M=10811$ and a polynomial $w(\bm{c},\bm{x}_{0,i},t)$ of degree $6$ over $t$ as a mathematical model to perform computations. Since $\bm{x}_{0,i}$ is known, $w(\bm{c},\bm{x}_{0,i},t)$ is of the form $\sum_{i=0}^6 c_i t^i$.  Note that the number $k+1$ of decision variables in \eqref{lp} is $8$ and consequently $M\geq 10811$ according to Theorem \ref{conclusion1}. 

 We obtain $\xi^*=0.33$ via solving the linear program \eqref{lp} with $U_c=U_{\xi}=100$. Therefore, we have that with confidence of at least $1-10^{-20}$,
\begin{equation}
    \label{11}
y_{\bm{x}_{0,i}}(t)\in [z_{\bm{x}_{0,i}}(t)-0.33, z_{\bm{x}_{0,i}}(t)+0.33]
\end{equation}
for all $t$ in  $[0,10]$ except at most a fraction of probability measure $0.01$, where $z_{\bm{x}_{0,i}}(\cdot):[0,T] \rightarrow \mathbb{R}$ is the trajectory of the mathematical model $z(t)=w(\bm{c}^*,\bm{x}_{0,i},t)$. We also take the time step $\Delta t=10^{-5}$ and the corresponding states $\big(y_{\bm{x}_{0,i}}(j\Delta t)\big)_{j=0}^{10^6}$ on the ground truth trajectory  to verify the satisfiability of  \eqref{11}, i.e., whether 
$y_{\bm{x}_{i,0}}(j\Delta t)\in [\bm{z}_{\bm{x}_{0,i}}(j\Delta t)-0.33, \bm{z}_{\bm{x}_{0,i}}(j\Delta t)+0.33]$
holds for $j\in \{0,1,\ldots, 10^{6}\}$. The satisfiability ratio is $100\%$.

Since $[z_{\bm{x}_{0,i}}(t)-0.33, z_{\bm{x}_{0,i}}(t)+0.33] \cap \mathtt{Uns}=\emptyset$ for $t\in [0,10]$, we have that the amount of time the system \eqref{bb} with the input $(1.4,2.3)^{\top}$ spends inside the unsafe set $\mathtt{Uns}$ does not exceed $0.1$, with confidence of at least $1-10^{-20}$.  

\oomit{Similarly, we apply the above procedure to the characterization of the trajectory $y_{\bm{x}_{0,i}}(\cdot): [0,10]\rightarrow \mathbb{R}$ when it describes the time evolution of the state $x_2$ in the system \eqref{ode1}, where $\beta=10^{-20}$, $\epsilon=0.01$ and $M=10811$ and a polynomial $w'(\bm{c},\bm{x}_{0,i},t)$ of degree $6$ is used. We obtain $\xi^*=0.85$ via solving the linear program \eqref{lp} with $U_c=U_{\xi}=100$. Thus, we have that with confidence of at least $1-10^{-20}$,
\begin{equation}
\label{111}
y_{\bm{x}_{0,i}}(t)\in [\bm{z}_{\bm{x}_{0,i}}(t)-0.85, \bm{z}_{\bm{x}_{0,i}}(t)+0.85]
\end{equation}
for all $t$ in  $[0,10]$ except at most a fraction $0.01$. 

Also, within the Monte-Carlo testing framework, we take the time step $\Delta t=10^{-5}$ and the corresponding states $\big(y_{\bm{x}_{0,i}}(j\Delta t)\big)_{j=0}^{10^6}$ on the ground truth trajectory to verify the satisfiability of  \eqref{11}, i.e., whether 
\[y_{\bm{x}_{0,i}}(j\Delta t)\in [\bm{z}_{\bm{x}_{0,i}}(j\Delta t)-0.85, \bm{z}_{\bm{x}_{0,i}}(j\Delta t)+0.85]\]
holds for $j\in \{0,1,\ldots, 10^{6}\}$. The satisfiability ratio of is $100\%$ as well.}

\begin{figure}
\center
   \includegraphics[width=3.05in,height=1.4in]{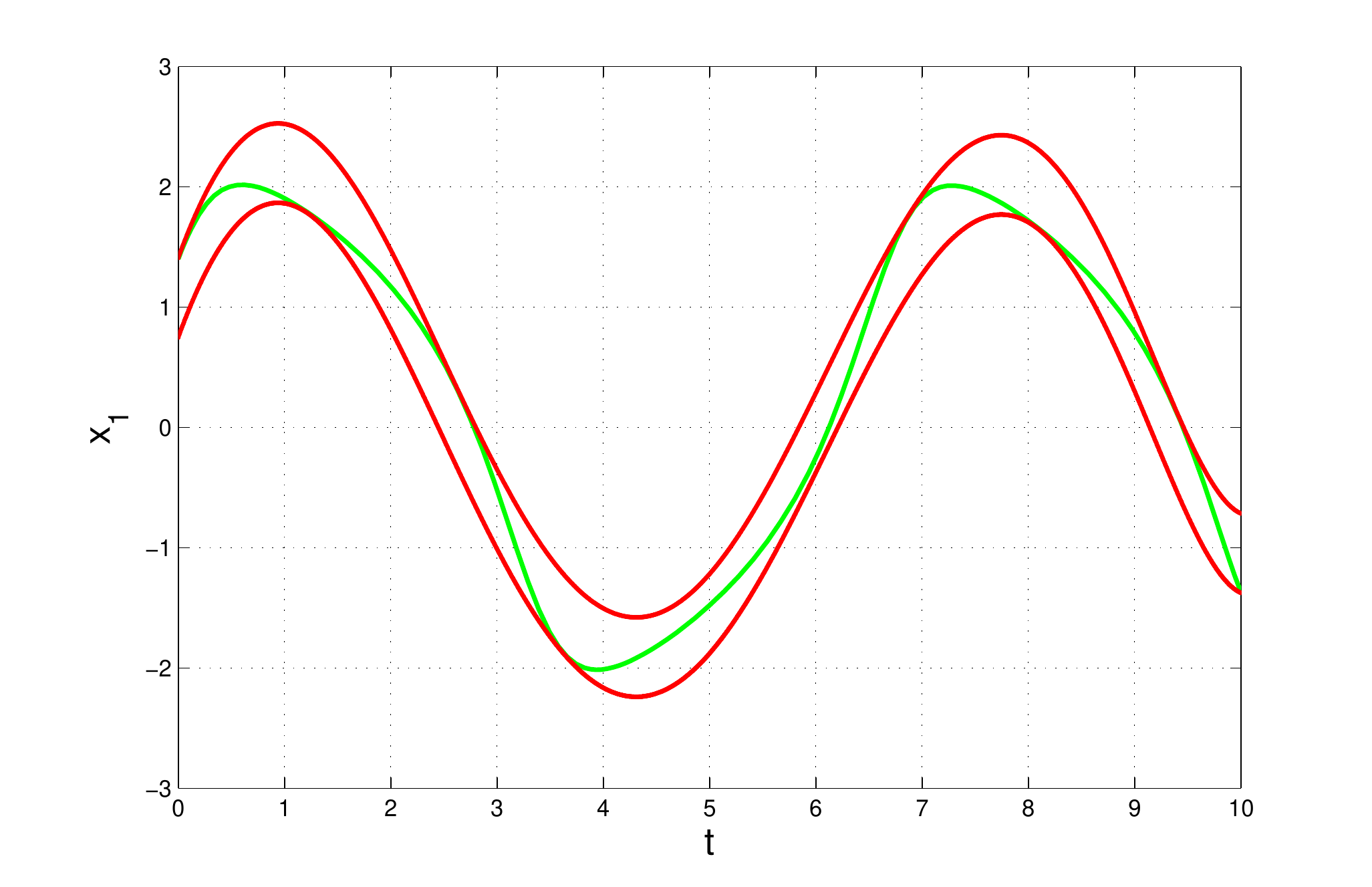}
     \caption{An illustration of the trajectory reachability for Example \ref{van}. The green curve denotes the ground truth trajectory. The red curve denotes $z_{\bm{x}_{0,i}}(\cdot)+\xi^*: [0,10]\rightarrow \mathbb{R}$ and $z_{\bm{x}_{0,i}}(\cdot)-\xi^*: [0,10]\rightarrow \mathbb{R}$ respectively.}
 \label{fig}
\end{figure}
\end{example} 
\subsubsection{{\rm Multiple Trajectories Verification}}
\label{MTC}
In Subsection 3.2.1 we considered one trajectory characterization of the system \eqref{bb}. In this subsection we extend the method in Subsection 3.2.1 to multiple trajectories characterization. These trajectories are the ones of the system \eqref{bb} with inputs $\bm{x}_{0,1},\ldots,\bm{x}_{0,N}$.

This extension is straightforward. We just  need to enrich the constraints in \eqref{lp} by incorporating these discrete-time trajectories $\Big((\bm{x}_{0,1},t_j,\bm{y}_{1,j})\Big)_{j=1}^M$, $\ldots$, $\Big((\bm{x}_{0,N},t_j,\bm{y}_{N,j})\Big)_{j=1}^M$, consequently resulting in the following linear program:
 \begin{equation}
 \label{lp1}
\begin{split}
&\min_{\bm{c},\xi} \xi\\
&\text{s.~t.~for each~} j=1,\ldots, M\text{~and~}i=1,\ldots,N:\\
& w(\bm{c}, \bm{x}_{0,i},t_j)-y_{i,j}\leq \xi,\\
& y_{i,j}-w(\bm{c}, \bm{x}_{0,i},t_j)\leq \xi,\\
&-U_c\leq c_{l}\leq U_c, l=1,\ldots,k,\\
& 0\leq \xi \leq U_{\xi},
\end{split}
 \end{equation}
where $U_c\in \mathbb{R}_{\geq 0}$ is a given upper bound for $c_{l}$, $l=1,\ldots,k$, and $U_{\xi}\in \mathbb{R}_{\geq 0}$ is a given upper bound for $\xi$. Denote the optimal solution to \eqref{lp1} by $(\bm{c}^{**},\xi^{**})$. 

 We denote the trajectory of the mathematical model $z(t)=w(\bm{c}^*,\bm{x}_{0},t)$ with the input $\bm{x}_0$ by $z_{\bm{x}_{0}}(\cdot): [0,T]\rightarrow \mathbb{R}$.  Similarly, we have the following theorem for the solution obtained via solving the linear program \eqref{lp1}.
 \begin{theorem}
\label{conclusion2}
Let $(\bm{c}^{**},\xi^{**})$ be an optimal solution to \eqref{lp1}, $\epsilon\in (0,1)$, $\beta\in (0,1)$ and 
\begin{equation}
\label{NLP1}
\epsilon\geq \frac{2}{M}(\ln \frac{1}{\beta}+k+1). 
\end{equation} 
Then for each input $\bm{x}_{0,i}$,  $i=1,\ldots,N$, we have that with at least $1-\beta$ confidence,  \[P_t(\{t\in [0,T]\mid |w(\bm{c}^{**}, \bm{x}_{0,i},t)-b(\bm{x}_{0,i},t)| \leq \xi^{**}\})\geq 1-\epsilon.\]
\end{theorem}
\begin{proof}
According to the scenario optimization in Subsection \ref{SO}, we have that with at least $1-\beta$ confidence,  \[
\begin{split}
    &P_t(\{t\in [0,T]\mid \wedge_{i=1}^N |w(\bm{c}^{**}, \bm{x}_{0,i},t)-b(\bm{x}_{0,i},t)| \leq \xi^{**}\})\\
    &\geq 1-\epsilon.
\end{split}\] Since 
\[\begin{split}
&P_t(\{t\in [0,T]\mid |w(\bm{c}^{**}, \bm{x}_{0,i},t)-b(\bm{x}_{0,i},t)| \leq \xi^{**}\})\geq \\
&P_t(\{t\in [0,T]\mid \wedge_{i=1}^M|w(\bm{c}^{**}, \bm{x}_{0,i},t)-b(\bm{x}_{0,i},t)| \leq \xi^{**}\})
\end{split}\]
for $i\in \{1,\cdots, M\}$, the conclusion follows directly.
\end{proof}

From Theorem \ref{conclusion2}, we have that for each trajectory $y_{\bm{x}_{0,i}}(\cdot):[0,T]\rightarrow \mathbb{R}$ of the system \eqref{bb} with the input $\bm{x}_{0,i}$,  $i=1,\ldots, N$, with confidence of at least $1-\beta$, \[y_{\bm{x}_{0,i}}(t)\in [z_{\bm{x}_{0,i}}(t)-\xi^{**}, z_{\bm{x}_{0,i}}(t)+\xi^{**}]\] for all $t$ in  $[0,T]$ but at most a fraction of probability measure $\epsilon$, i.e., with confidence of at least $1-\beta$, each of the $N$ trajectories of the system \eqref{bb}  deviates from the corresponding one of the mathematical model $\bm{z}(t)=w(\bm{c}^{**},\bm{x}_0,t)$ by at most $\xi^{**}$ for all $t\in [0,T]$ but at most a fraction $\epsilon$.

Consequently, the solution to Problem \ref{veri} for the system \eqref{bb} with multiple trajectories is presented below:

\textit{
If $[z_{\bm{x}_{0,i}}(t)-\xi^{**}, z_{\bm{x}_{0,i}}(t)+\xi^{**}]$ does not intersect the unsafe set $\mathtt{Uns}$ for $t\in [0,T]$, $i\in \{1,\ldots,N\}$, we have that the amount of time the system \eqref{bb} with the input $\bm{x}_{0,i}$ spends inside the unsafe set $\mathtt{Uns}$ does not exceed $\epsilon T$, with confidence of at least $1-\beta$.  
 }

It is worth remarking that the family of inputs $(\bm{x}_{0,i})_{i=1}^N$ here does not require to be extracted independently according to the probability distribution $P_{\bm{x}_0}$. They can be arbitrary $N$ inputs of interest in the set $\mathcal{X}_0$.  

\begin{example}
\label{van1}
Let's take the system in Example \ref{van} as an instance to illustrate the case of two trajectories verification. These two trajectories, which are presented in Fig. \ref{fig3},  respectively describe the time evolution of the state $x_1$ in \eqref{ode1} with two different inputs $\bm{x}_{0,1}=(1.25,2.28)^{\top}$ and $\bm{x}_{0,2}=(1.55,2.32)^{\top}$.

Let $\beta=10^{-20}$ and $\epsilon=0.01$. In this example we use $M=26211$ and a polynomial $w(\bm{c},\bm{x}_0,t)$ of degree $6$ as a mathematical model, which is input-dependent and is linear in $\bm{c}$, to perform computations. The number $k+1$ of decision variables in \eqref{lp} is $85$ and thus $M\geq 26211$ from Theorem \ref{conclusion1}.

We obtain $\xi^{**}=0.34$ via solving the linear program \eqref{lp1} with $U_c=U_{\xi}=100$. Thus,  for each $i=1,2$, we have that with confidence of at least $1-10^{-20}$, $y_{\bm{x}_{0,i}}(t)\in [z_{\bm{x}_{0,i}}(t)-0.34, z_{\bm{x}_{0,i}}(t)+0.34]$ for all $t\in [0,10]$ except a small fraction $0.01$, where $z_{\bm{x}_{0,i}}(\cdot):[0,T] \rightarrow \mathbb{R}$ is the trajectory of the mathematical model $z(t)=w(\bm{c}^{**},\bm{x}_{0,i},t)$. Like Example \ref{van}, within the Monte-Carlo testing framework, we take the time step $\Delta t= 10^{-5}$ and the corresponding states $\big(y_{\bm{x}_{0,i}}(j\Delta t)\big)_{j=0}^{10^6}$ on the  ground truth trajectory with the input $\bm{x}_{0,i}$ to verify whether $y_{\bm{x}_{0,i}}(j\Delta t)\in [z_{\bm{x}_{0,i}}(j\Delta t)-0.34, z_{\bm{x}_{0,i}}(j\Delta t)+0.34]$ for $j\in \{0,1,\ldots,10^6\}$, where $i=1,2$. The satisfiability ratio is $100\%$ for both of these two trajectories.

Since $[z_{\bm{x}_{0,i}}(t)-0.34, z_{\bm{x}_{0,i}}(t)+0.34] \cap \mathtt{Uns}=\emptyset$ for $t\in [0,10]$ and $i=1,2$, we have that the amount of time the system \eqref{bb} with each of the two inputs $\bm{x}_{0,1}=(1.25,2.28)^{\top}$ and $\bm{x}_{0,2}=(1.55,2.32)^{\top}$ spends inside the unsafe set $\mathtt{Uns}$ does not exceed $0.1$, with confidence of at least $1-10^{-20}$.

\begin{figure}
\center
   \includegraphics[width=3.05in,height=1.1in]{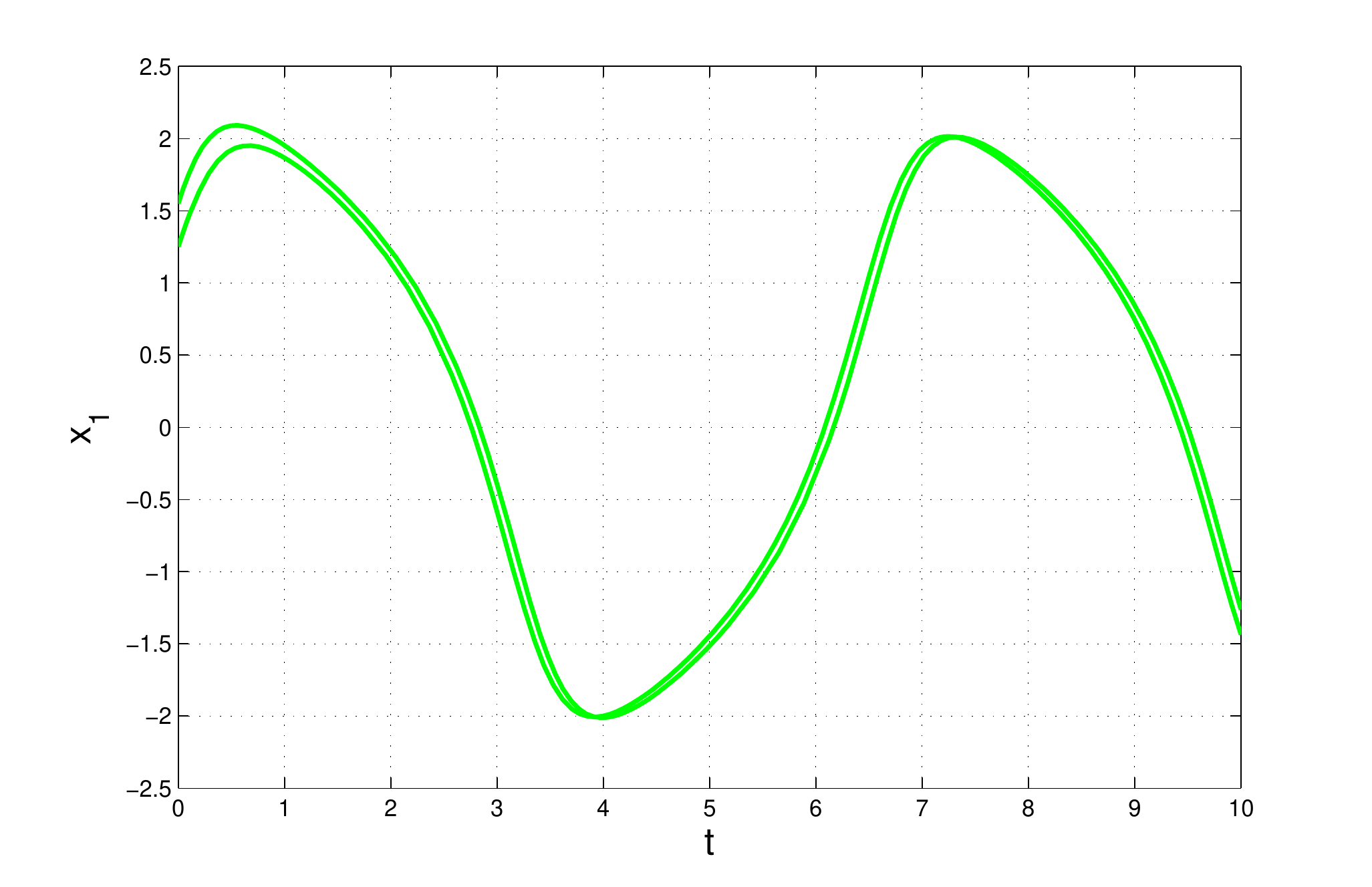}
   \includegraphics[width=3.05in,height=1.1in]{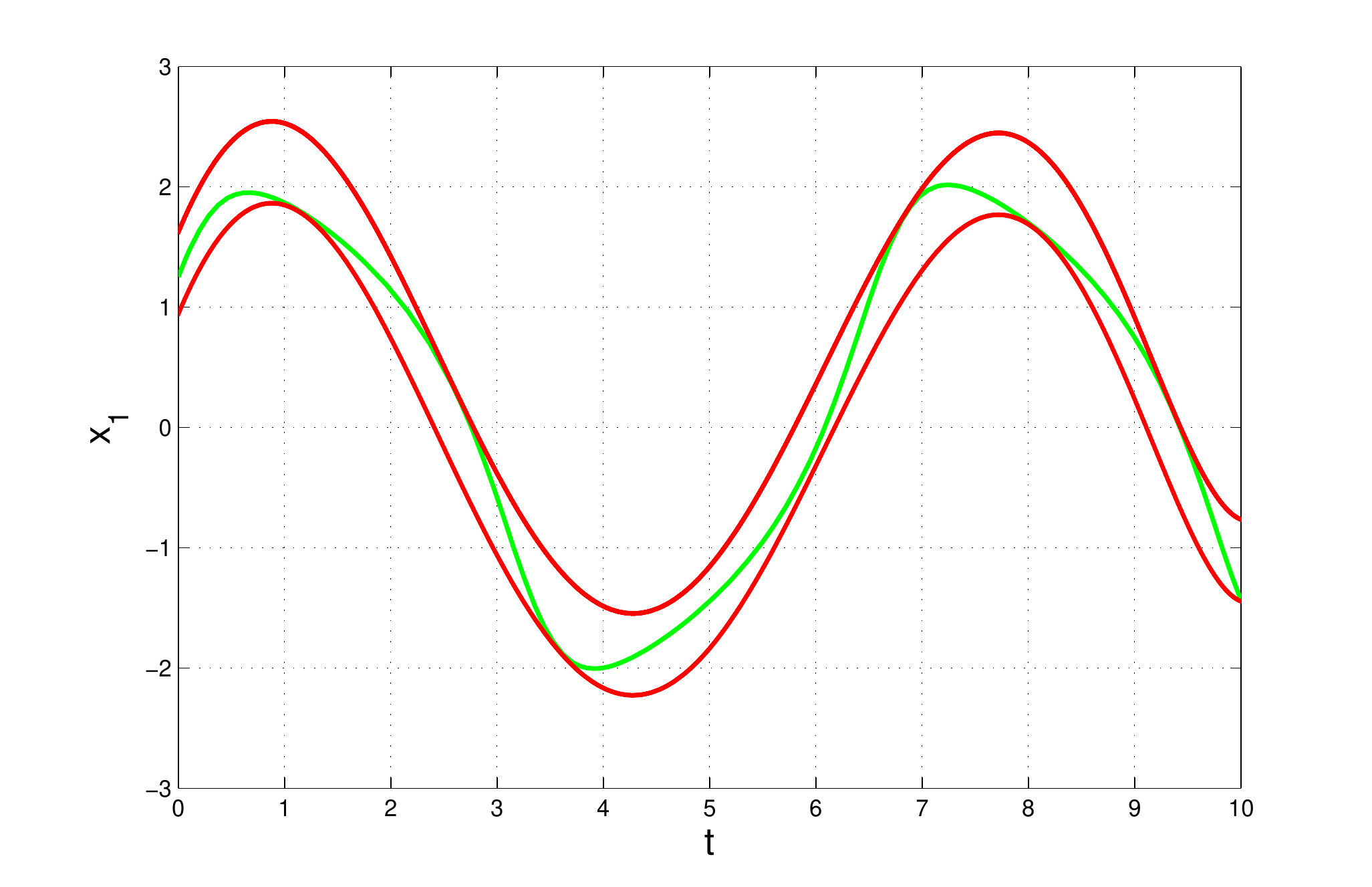}
   \includegraphics[width=3.05in,height=1.1in]{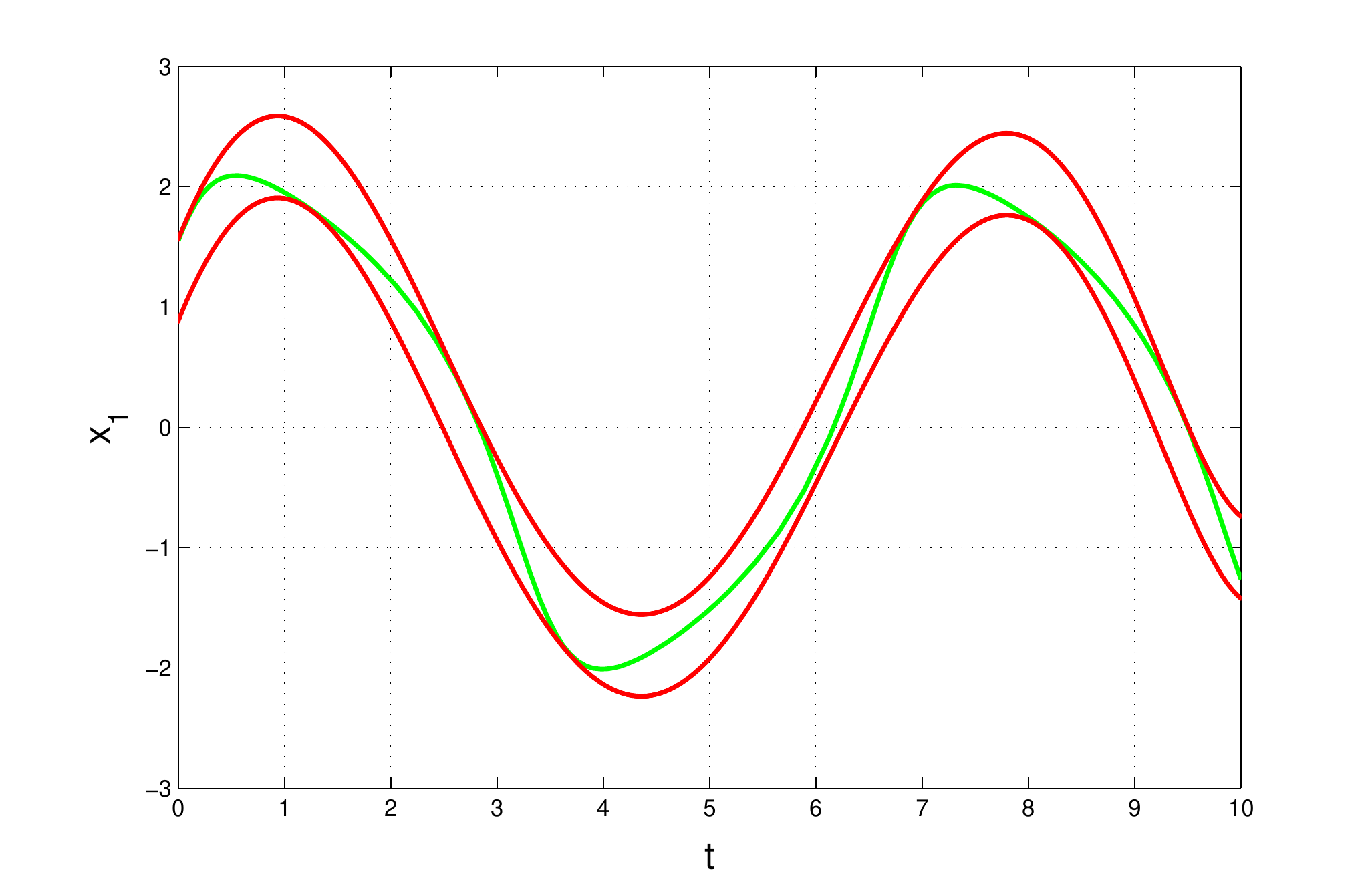}
     \caption{An illustration of two trajectories reachability for Example \ref{van1} with inputs $\bm{x}_{0,1}=(1.25,2.28)^{\top}$ and $\bm{x}_{0,2}=(1.55,2.32)^{\top}$. The green curves denote the two ground truth trajectories. From middle to right ($i=1,2$): the green curve denotes  $y_{\bm{x}_{0,i}}(\cdot):[0,T]\rightarrow \mathbb{R}$, and the red curves correspond to $z_{\bm{x}_{0,i}}(\cdot)+\xi^{**}: [0,T]\rightarrow \mathbb{R}$ and $z_{\bm{x}_{0,i}}(\cdot)-\xi^{**}: [0,T]\rightarrow \mathbb{R}$ respectively.}
     \label{fig3}
\end{figure}
\end{example}

\subsubsection{{\rm All Trajectories Verification}}
In this subsection we further extend the method in Subsection 3.2.2 for multiple trajectories verification to all trajectories verification of the system \eqref{bb} with the input set $\mathcal{X}_0$. Unlike in Subsection 3.2.2, the family of inputs $(\bm{x}_i)_{i=1}^N$ in this situation should be extracted independently according to the probability distribution $P_{\bm{x}_0}$.

 \begin{theorem}
\label{conclusion3}
Let $(\bm{c}^{**},\xi^{**})$ be an optimal solution to \eqref{lp1}, $\epsilon_1\in (0,1)$, $\beta_1\in (0,1)$, $\epsilon_2\in (0,1)$, $\beta_2\in (0,1)$,  and 
\begin{align}
&\epsilon_1\geq \frac{2}{M}(\ln \frac{1}{\beta_1}+k+1), \label{NLP21}\\
&\epsilon_2\geq \frac{2}{N}(\ln \frac{1}{\beta_2}+k+1). \label{NLP22}
\end{align}

Then we have that with at least $1-\beta_2$ confidence,  
$P_{\bm{x}_0}(\{\bm{x}_0 \mid \bm{x}_0\in \mathcal{X}\}) \geq 1-\epsilon_2$, where $\mathcal{X}=$
\[\begin{split}
&\left\{\bm{x}_0\in \mathcal{X}_0\middle|\;
\begin{aligned}
&P_t\Bigg(\left\{t\in [0,T]\middle|\;
\begin{aligned}
&|w(\bm{c}^{**}, \bm{x}_{0},t)-b(\bm{x}_{0},t)|\\
&\leq \xi^{**}
\end{aligned}
\right\}\Bigg)\\
&\geq 1-\epsilon_1, \text{with confidence of at least~}1-\beta_1.
\end{aligned}
\right\}.
\end{split}
\]
\end{theorem}
\begin{proof}
Let us fix the time instances $t_1,\cdots,t_M$ firstly, we have that with confidence of at least $1-\beta_2$,  \[
\begin{split}
&P_{\bm{x}_0}\Bigg(\left\{\bm{x}_0\in \mathcal{X}_0\middle|\; 
\bigwedge_{j=1}^M |w(\bm{c}^{**},\bm{x}_0,t_j)-b(\bm{x}_0,t_j)|\leq \xi^{**}\right\}\Bigg)\\
&\geq 1-\epsilon_2
\end{split}
\]

Let $\tilde{\mathcal{X}}_0=\{\bm{x}_0\in \mathcal{X}_0\mid  \wedge_{j=1}^M |w(\bm{c}^{**},\bm{x}_0,t_j)-b(\bm{x}_0,t_j)|\leq \xi^{**}\}$. Obviously, $\bm{x}_{0,i} \in \tilde{\mathcal{X}}_0$, $i=1,\ldots,N$.  For $\bm{x}_0\in \tilde{\mathcal{X}}_0$, we can add the constraints involving $\bm{x}_0$ to the linear program \eqref{lp1} and obtain the following linear program:
\begin{equation}
 \label{lp4}
\begin{split}
&\min_{\bm{c},\xi} \xi\\
&\text{s.~t.~for each~} j=1,\ldots, M\text{~and~}i=1,\ldots,N:\\
& w(\bm{c}, \bm{x}_{0,i},t_j)-y_{i,j}\leq \xi,\\
& y_{i,j}-w(\bm{c}, \bm{x}_{0,i},t_j)\leq \xi,\\
& w(\bm{c}, \bm{x}_{0},t_j)-b(\bm{x}_0,t_j)\leq \xi,\\
& b(\bm{x}_0,t_j)-w(\bm{c}, \bm{x}_{0},t_j)\leq \xi,\\
&-U_c\leq c_{l}\leq U_c, l=1,\ldots,k,\\
& 0\leq \xi \leq U_{\xi}.
\end{split}
\end{equation}
Obviously, $(\bm{c}^{**},\xi^{**})$ is also an optimal solution to \eqref{lp4}. Since the time instances $t_1,\cdots,t_M$ are also extracted  independently according to the distribution $P_t$, Theorem \ref{conclusion2} indicates that with confidence of at least $1-\beta_1$, 
\[P_t(\{t\in [0,T]\mid |w(\bm{c}^{**}, \bm{x}_{0},t)-b(\bm{x}_{0},t)| \leq \xi^{**}\})\geq 1-\epsilon_1\]
for $\bm{x}_0\in \tilde{\mathcal{X}}_0$. Thus, we have $\tilde{\mathcal{X}}_0\subseteq \mathcal{X}$ and consequently the conclusion follows.
\end{proof}

From Theorem \ref{conclusion3}, we have that with confidence of at least $1-\beta_2$, the probability measure of the set $\mathcal{X}$ is larger than $1-\epsilon_2$.  The set $\mathcal{X}$ is a set of inputs such that the trajectory of the system \eqref{bb} with each of them does not deviate from the corresponding one of the model $z(t)=w(\bm{c}^{**},\cdot,\cdot): \mathbb{R}^n\times [0,T]\rightarrow \mathbb{R}$ by $\xi^{**}$ for all $t\in [0,T]$ but at most a fraction $\epsilon_1$.

Thus, the solution to Problem \ref{veri} for the system \eqref{bb} with all trajectories originating from the set $\mathcal{X}_0$ is presented below:

\textit{
If $[z_{\bm{x}_{0}}(t)-\xi^{**},z_{\bm{x}_{0}}(\cdot)+\xi^{**}]\cap \mathtt{Uns}=\emptyset$ for $\bm{x}_0\in \mathcal{X}_0$ and $t\in [0,T]$, we have that with confidence of at least $1-\beta_2$, the probability measure of inputs in $\mathcal{X}_0$ such that the amount of time the system \eqref{bb} with each of them spends inside $\mathtt{Uns}$ does not exceed $\epsilon_1 T$ with confidence of at least $1-\beta_1$, is larger than $1-\epsilon_2$.
}

Although the size of the linear program \eqref{lp1} for computing PAC models does not depend on the dimension of the system \eqref{bb}, it heavily depends on $\epsilon_1,\beta_1,\epsilon_2,\beta_2$ and the number of unknown parameters in a pre-specified PAC model template according to inequalities \eqref{NLP21} and \eqref{NLP22} in Theorem \ref{conclusion3}.
\begin{example}
\label{ex3}
Let's take the system in Example \ref{van} again as an instance to illustrate the case of all trajectories characterization. The input set is assumed to be  $\mathcal{X}_0=[1.25,1.55]\times [2.28,2.32]$.

Let $\beta_1=10^{-10}$, $\epsilon_1=0.3$, $\beta_2=10^{-10}$ and $\epsilon_2=0.5$. In this example we use $M=207$, $N=125$ and a polynomial $w(\bm{c},t)$ of degree $6$ as a mathematical model, which is input-independent and is linear in $\bm{c}$, to perform computations. The number $k+1$ of decision variables in \eqref{lp1} is $8$ and consequently $M\geq 207$ and $N\geq 125$ according to Theorem \ref{conclusion3}. The computation time for solving the resulting linear program is $150.32$ seconds. The reason that an input-independent model is used  is to reduce the number of decision variables in \eqref{lp1}, which further results in reduction of the size of extracted samples according to inequalities \eqref{NLP21} and \eqref{NLP22} and thus reduction of the size of the linear program \eqref{lp1}. These computations were performed on an i7-7500U 2.70GHz CPU with 32G RAM running Windows 10. 

We obtain $\xi^{**}=0.38$ via solving the linear program \eqref{lp1} with $U_c=U_{\xi}=100$. Therefore, with confidence of at least $1-10^{-10}$, the probability measure of inputs in $\mathcal{X}_0$ such that with confidence of at least $1-10^{-10}$,
\begin{equation}
    \label{13}
y_{\bm{x}_0}(t)\in [z_{\bm{x}_0}(t)-0.38,z_{\bm{x}_0}(t)+0.38]
\end{equation}
for all $t\in  [0,10]$ but at most a fraction $0.3$, is larger than $0.5$, where $z_{\bm{x}_{0}}(\cdot):[0,T] \rightarrow \mathbb{R}$ is the trajectory of the mathematical model $z(t)=w(\bm{c}^{**},t)$. Within the Monte-Carlo testing framework, we extract $10^4$ inputs $(\bm{x}'_{i,0})_{i=1}^{10^4}$ from $\mathcal{X}_0$ independently according to the probability distribution $P_{\bm{x}_0}$ and then obtain their corresponding ground truth trajectories for validating the above conclusion. Like Example \ref{van}, we take the time step $\Delta t= 10^{-5}$ and the states $\big(y_{\bm{x}'_{0,i}}(j\Delta t)\big)_{j=0}^{10^6}$ on the ground truth trajectory with the input $\bm{x}'_{0,i}$ to verify the satisfiability of \eqref{13}, where $i=1,\ldots,10^{4}$. The satisfiability ratio of $10^4$ inputs such that
\begin{equation*} 
y_{\bm{x}'_{0,i}}(j\Delta t)\in [z_{\bm{x}'_{0,i}}(j\Delta t)-0.38, z_{\bm{x}'_{0,i}}(j\Delta t)+0.38]
\end{equation*}
for all $j\in \{0,\ldots,10^6\}$ but at most a fraction $0.05$, is $100\%$.

Since $[z_{\bm{x}_{0}}(t)-\xi^{**},z_{\bm{x}_{0}}(\cdot)+\xi^{**}]\cap \mathtt{Uns}=\emptyset$ for $\bm{x}_0\in \mathcal{X}_0$ and $t\in [0,10]$, we have that with at least $1-10^{-10}$ confidence, the probability measure of inputs in $\mathcal{X}_0$ such that the amount of time the system \eqref{bb}  with each of them spends inside $\mathtt{Uns}$ does not exceed $3$ with at least $1-10^{-10}$ confidence, is larger than $0.5$.

\begin{figure}
\center
   \includegraphics[width=3.05in,height=1.3in]{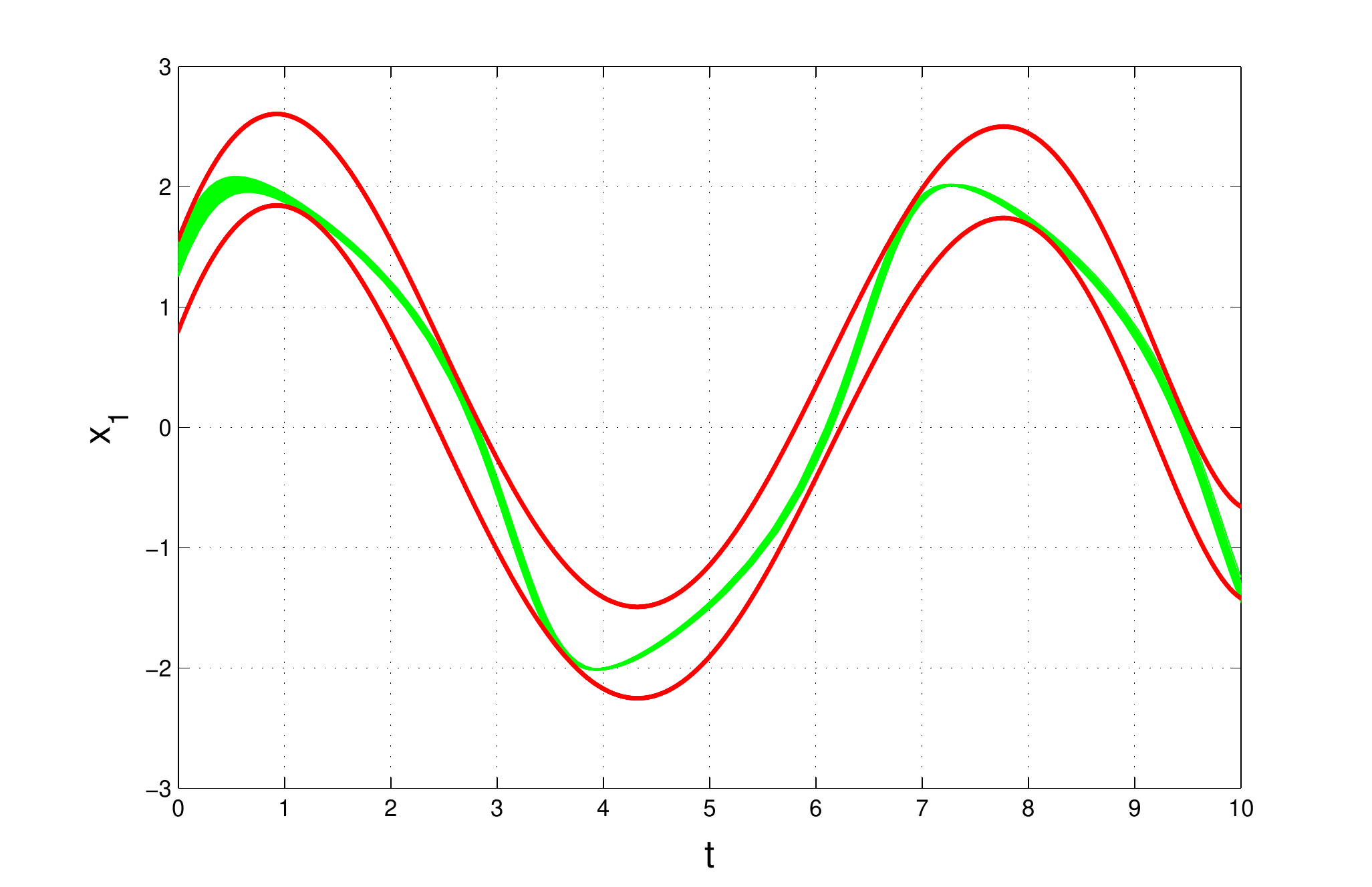}
    \caption{An illustration of all trajectories reachability for Example \ref{ex3} with $\mathcal{X}_0=[1.25,1.55]\times [2.28,2.32]$. The green curves denote the trajectories generated by the extracted $N$ inputs. The red curves denote $w(\bm{c}^{**},\cdot)-\xi^{**}:[0,T]\rightarrow \mathbb{R}$ and $w(\bm{c}^{**},\cdot)+\xi^{**}:[0,T]\rightarrow \mathbb{R}$ respectively. \oomit{The blue points denote the extracted $N$ inputs.}}
  \end{figure}
  \begin{figure}
  \centering
   \includegraphics[width=3.05in,height=1.3in]{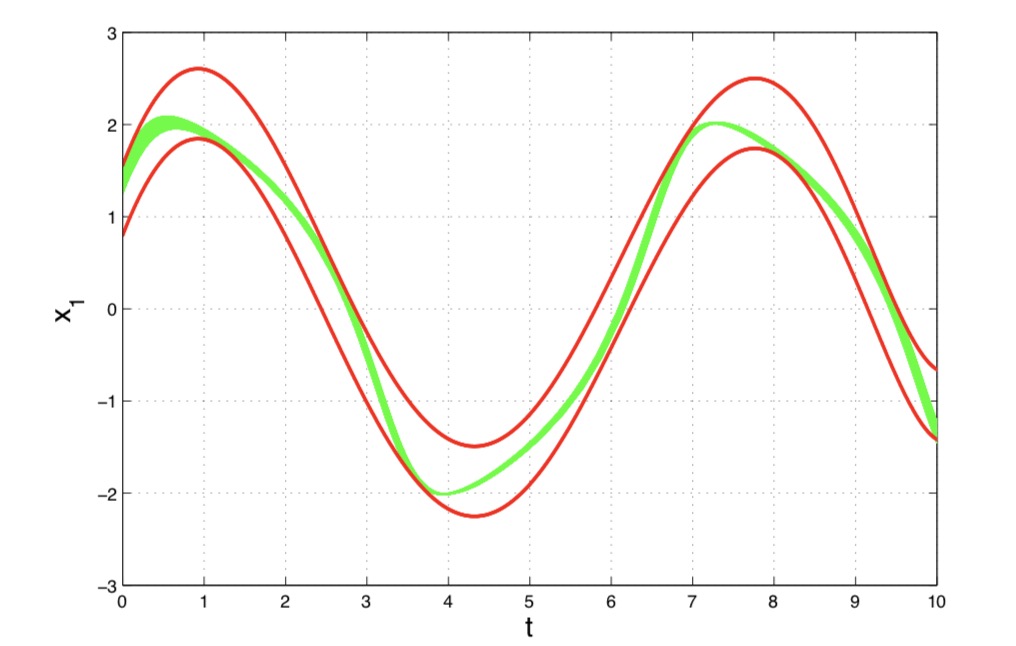}
    \caption{An illustration of Monte Carlo validation for Example \ref{ex3}. The green curves denote the extracted $10^4$ trajectories, and the red curves denote $w(\bm{c}^{**},\cdot)+\xi^{**}: [0,T]\rightarrow \mathbb{R}$ and $w(\bm{c}^{**},\cdot)-\xi^{**}: [0,T]\rightarrow \mathbb{R}$ respectively. \oomit{The blue points denote the extracted $10^4$ inputs.} }
    \end{figure}
\end{example}

\section{Experiments}
\label{experiments}
In this section we demonstrate the performance of our approach on three examples. All computations were performed on an i7-7500U 2.70GHz CPU with 32G RAM running Windows 10.

\begin{example}
\label{nine}
In this example we consider a black-box system of the form \eqref{bb} with $T=10$, $\mathcal{X}_0=[1.0,1.1]^9$ and $\mathtt{Uns}=\{y\in \mathbb{R}\mid y\leq -3\}$, which describes the time evolution of the state $x_1$ in  the following 9-dimensional biological model \cite{chen2015}: 
\begin{equation*}
\label{ninea}
\left\{
\begin{aligned}
&\dot{x}_1(t)=3x_3(t)-x_1(t)x_6(t), \dot{x}_2(t)=x_4(t)-x_2(t)x_6(t),\\
&\dot{x}_3(t)=x_1(t)x_6(t)-3x_3(t), \dot{x}_4(t)=x_2(t)x_6(t)-x_4(t),\\
&\dot{x}_5(t)=3x_3(t)+5x_1(t)-x_5(t),\\
&\dot{x}_6(t)=5x_5(t)+3x_3(t)+x_4(t)\\
&~~~~~~~~~~~~~~~~~~-x_6(t)(x_1(t)+x_2(t)+2x_8(t)+1),\\
&\dot{x}_7(t)=5x_4(t)+x_2(t)-0.5x_7(t),\\ &\dot{x}_8(t)=5x_7(t)-2x_6(t)x_8(t)+x_9(t)-0.2x_8(t),\\
&\dot{x}_9(t)=2x_6(t)x_8(t)-x_9(t).
\end{aligned}
\right.
\end{equation*}

Let $\epsilon_1=0.2$, $\beta_1=10^{-10}$, $\epsilon_2=0.3$ and $\beta_2=10^{-10}$. In this example we compute two polynomial models of degree 2 and 5 to illustrate our method.

1). We use $M=271$, $N=181$ and a polynomial $w(\bm{c},t)$ of degree $2$ as a mathematical model, which is input-independent and is linear in $\bm{c}$, to perform computations. Note that the number $k+1$ of decision variables in \eqref{lp1} is $4$ and consequently $M\geq 271$ and $N\geq 181 $ according to Theorem \ref{conclusion3}.  Via solving \eqref{lp1} with $U_c=U_{\xi}=100$ we obtain $\xi^{**}=0.17$. The computation time is $167.43$ seconds. Therefore, according to Theorem \ref{conclusion3}, we conclude that with at least $1-10^{-10}$ confidence, the probability measure of inputs in $\mathcal{X}_0$ such that with confidence of at least $1-10^{-10}$, 
\[y_{\bm{x}_0}(t)\in [z_{\bm{x}_{0}}(t)-0.17, z_{\bm{x}_{0}}(t)+0.17]\]
for all $t\in [0,10]$ but at most a fraction $0.2$, is larger than $0.7$, where $z_{\bm{x}_{0}}(\cdot):[0,T] \rightarrow \mathbb{R}$ is the trajectory of the model $z(t)=w(\bm{c}^{**},t)$. The reachability analysis is illustrated in Fig. \ref{41}.  Like Example \ref{ex3}, within the Monte-Carlo framework, we also extract $10^4$ inputs $(\bm{x}'_{i,0})_{i=1}^{10^4}$ to verify the conclusion, and obtain that the ratio of $10^4$ inputs such that
$y_{\bm{x}'_{i,0}}(j\Delta t)\in [z_{\bm{x}'_{i,0}}(j\Delta t)-0.17, z_{\bm{x}'_{i,0}}(j\Delta t)+0.17]$
for all $j\in \{0,\ldots,10^6\}$ but at most a fraction $0.05$, is larger than $97.87\%$, where $\Delta t=10^{-5}$.

Since $[z_{\bm{x}_{0}}(t)-0.17, z_{\bm{x}_{0}}(t)+0.17] \cap \mathtt{Uns}=\emptyset$ for $t\in [0,10]$ and $\bm{x}_0\in \mathcal{X}_0$, we have that with at least $1-10^{-10}$ confidence, the probability measure of inputs in $\mathcal{X}_0$ such that the amount of time the system \eqref{bb} with each of them spends inside $\mathtt{Uns}$ does not exceed $2$ with confidence of at least $1-10^{-10}$, is larger than $0.7$. 

\begin{figure}
\centering
   \includegraphics[width=3.05in,height=1.3in]{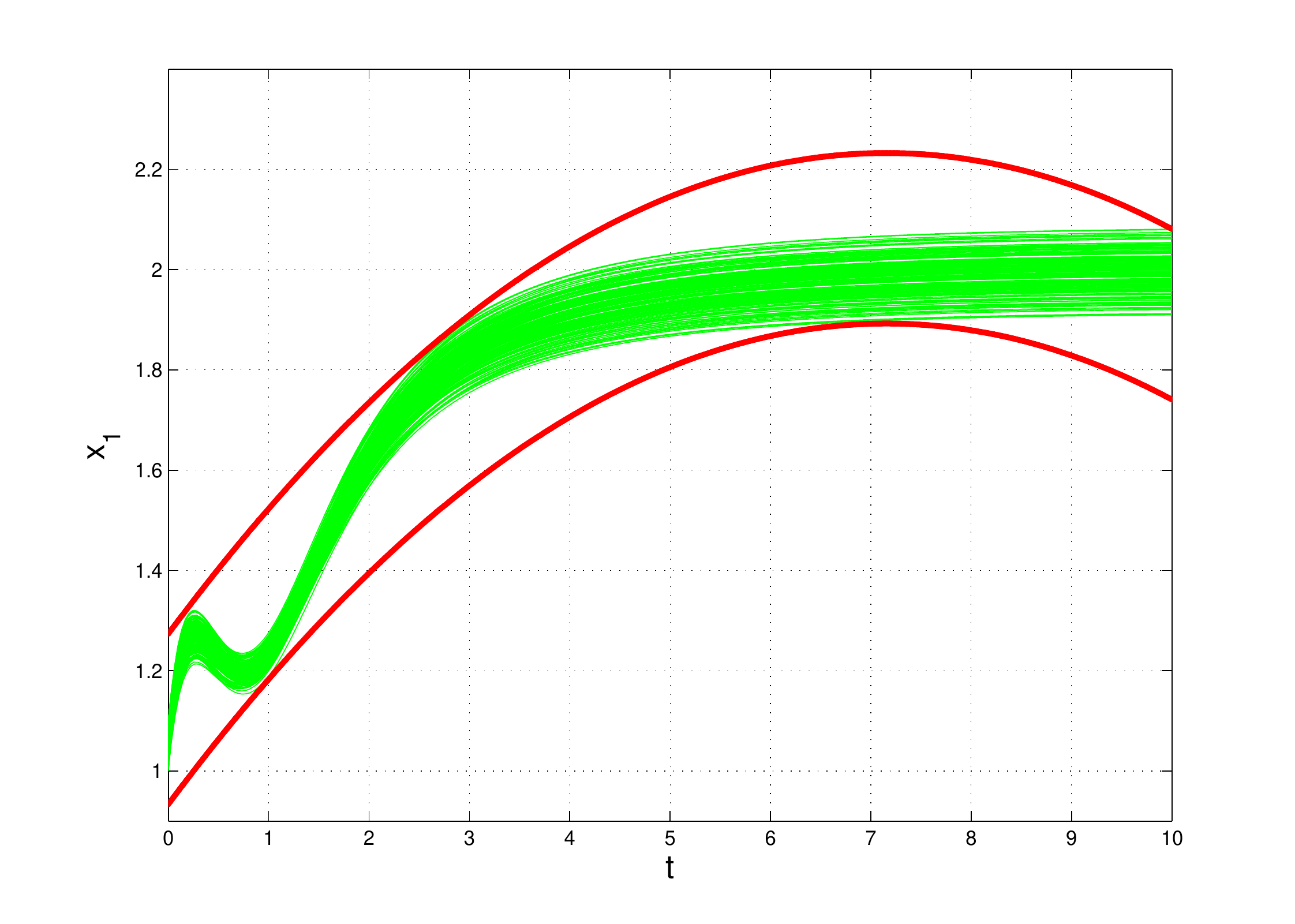}
  \caption{An illustration of trajectories reachability for Example \ref{nine} with the polynomial PAC model of degree 2. The green curves denote the extracted $181$ trajectories. The red curves denote $w(\bm{c}^{**},\cdot)-\xi^{**}:[0,T]\rightarrow \mathbb{R}$ and $w(\bm{c}^{**},\cdot)+\xi^{**}:[0,T]\rightarrow \mathbb{R}$ respectively. \oomit{The blue points denote the $181$ inputs in the $x_1-x_2$ plane.} }
  \label{41}
  \end{figure}

2). We use $M=301$, $N=201$ and a polynomial $w(\bm{c},t)$ of degree $5$ as a mathematical model, which is input-independent and is linear in $\bm{c}$, to perform computations. Note that the number $k+1$ of decision variables in \eqref{lp1} is $7$ and consequently $M\geq 301$ and $N\geq 201 $ according to Theorem \ref{conclusion3}.  Via solving \eqref{lp1} with $U_c=U_{\xi}=100$ we obtain $\xi^{**}=0.12$. The computation time is $223.83$ seconds.  Therefore,  according to Theorem \ref{conclusion3}, we conclude that with at least $1-10^{-10}$ confidence, the probability measure of inputs in $\mathcal{X}_0$ such that with confidence of at least $1-10^{-10}$, 
\[y_{\bm{x}_0}(t)\in [z_{\bm{x}_{0}}(t)-0.12, z_{\bm{x}_{0}}(t)+0.12]\]
for all $t\in [0,10]$ but at most a fraction $0.2$, is larger than $0.7$, where $z_{\bm{x}_{0}}(\cdot):[0,T] \rightarrow \mathbb{R}$ is the trajectory of the model $z(t)=w(\bm{c}^{**},t)$. The reachability analysis is illustrated in Fig. \ref{42}.  Within the Monte-Carlo framework we use the $10^4$ inputs $(\bm{x}'_{i,0})_{i=1}^{10^4}$ in the first case to verify the conclusion, and obtain that the ratio of $10^4$ inputs such that
$y_{\bm{x}'_{i,0}}(j\Delta t)\in [z_{\bm{x}'_{i,0}}(j\Delta t)-0.12, z_{\bm{x}'_{i,0}}(j\Delta t)+0.12]$
for all $j\in \{0,\ldots,10^6\}$ but at most a fraction $0.05$, is larger than $98.56\%$, where $\Delta t=10^{-5}$.

Similarly, due to the fact that $[z_{\bm{x}_{0}}(t)-0.12, z_{\bm{x}_{0}}(t)+0.12] \cap \mathtt{Uns}=\emptyset$ for $t\in [0,10]$ and $\bm{x}_0\in \mathcal{X}_0$, we have that with at least $1-10^{-10}$ confidence, the probability measure of inputs in $\mathcal{X}_0$ such that the amount of time the system \eqref{bb} with each of them spends inside the unsafe set $\mathtt{Uns}$ does not exceed $2$ with confidence of at least $1-10^{-10}$, is larger than $0.7$. 

From the comparison results illustrated in Fig. \ref{fig43} for the above two cases with the same {\rm PAC} guarantees, i.e., $\epsilon_1$, $\epsilon_2$, $\beta_1$ and $\beta_2$ are the same, we observe that polynomial models of higher degree could describe the internal dynamics of the system \eqref{bb} more exactly, but with more computation time.

\begin{figure}
\centering
   \includegraphics[width=3.05in,height=1.3in]{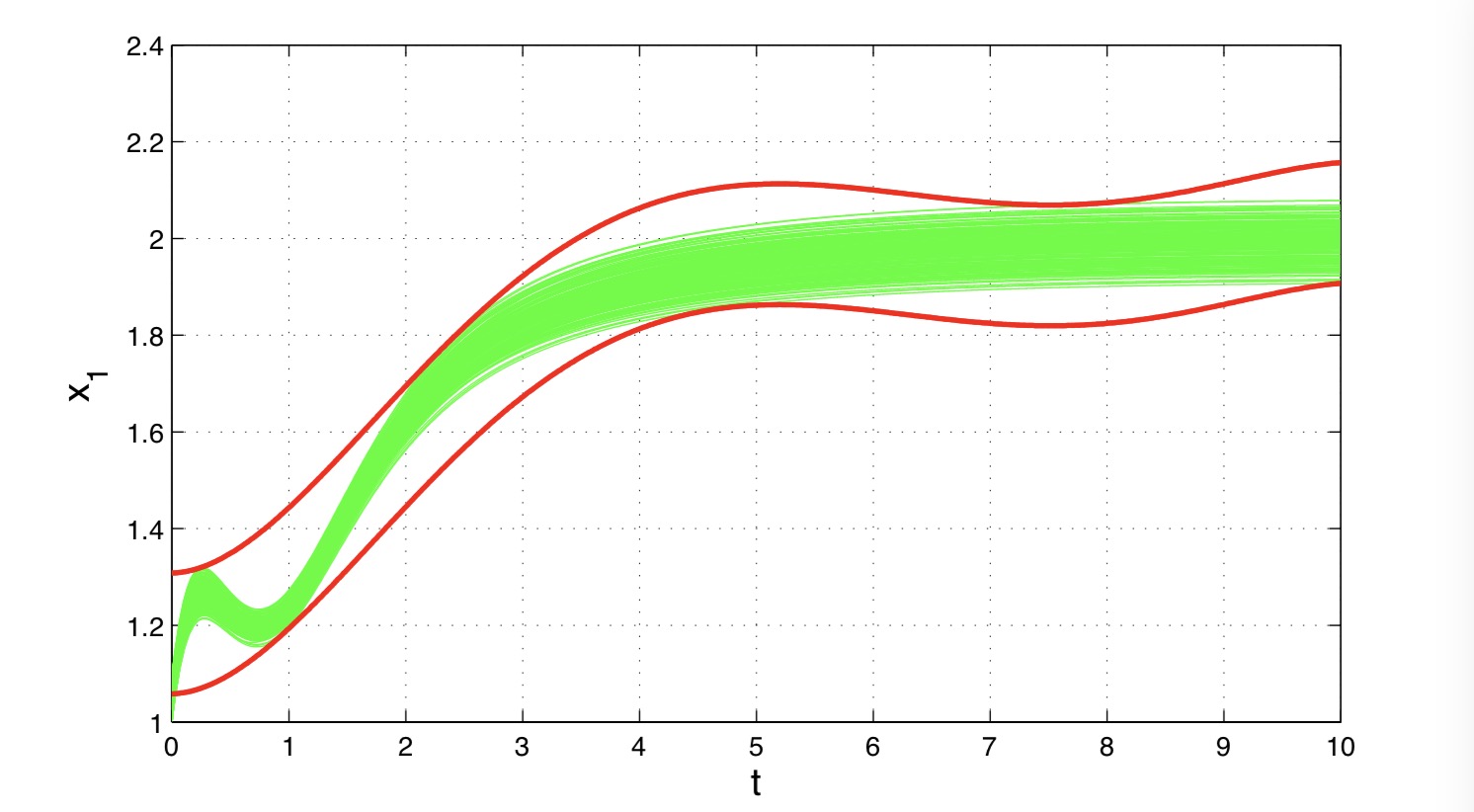}
  \caption{An illustration of trajectories reachability for Example \ref{nine} with the polynomial model of degree 5. The green curves denote the extracted $201$ trajectories. The red curves denote $w(\bm{c}^{**},\cdot)-\xi^{**}:[0,T]\rightarrow \mathbb{R}$ and $w(\bm{c}^{**},\cdot)+\xi^{**}:[0,T]\rightarrow \mathbb{R}$ respectively. \oomit{The blue points denote the $201$ inputs in the $x_1-x_2$ plane.} }
  \label{42}
  \end{figure}
  \begin{figure}
  \centering
   \includegraphics[width=3.05in,height=1.3in]{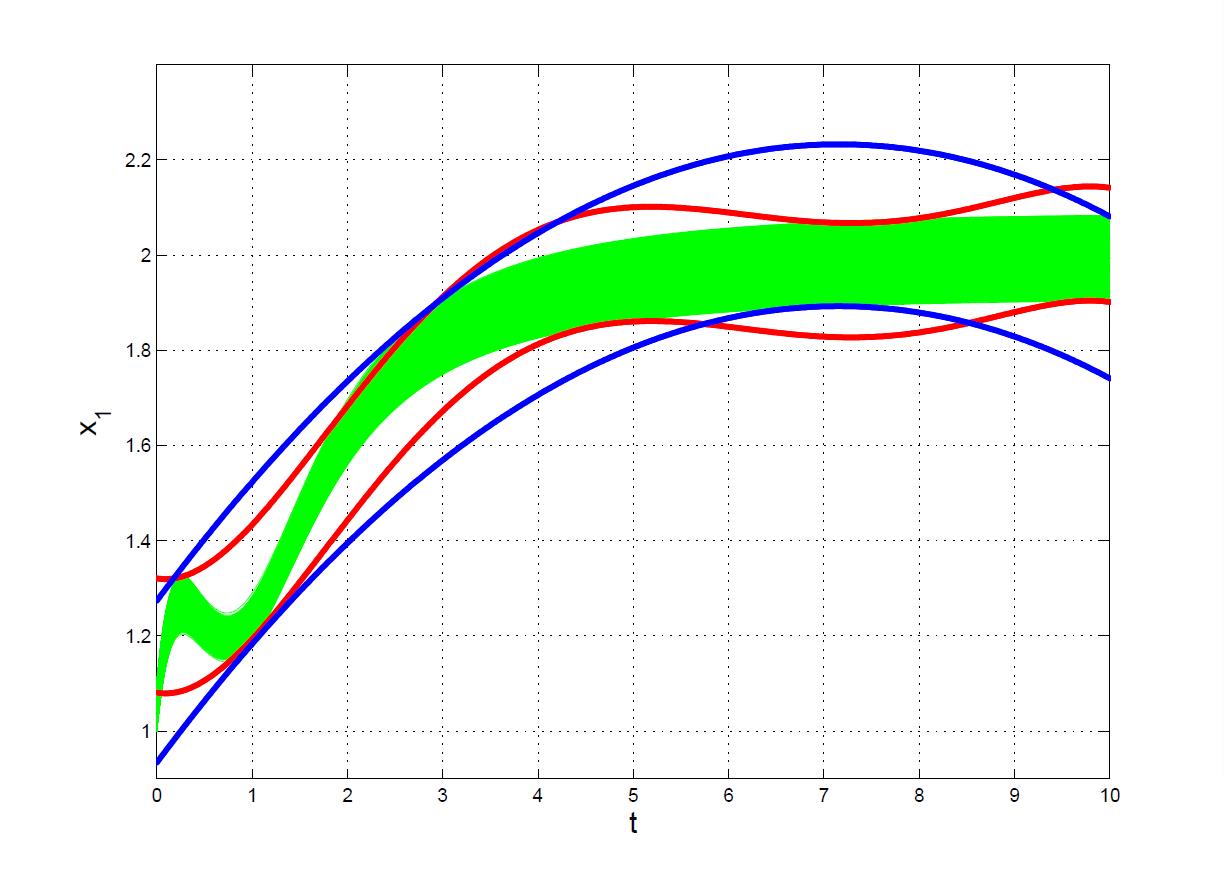}
     \caption{An illustration of Monte Carlo validation for Example \ref{nine}. The green curves denote the extracted $10^4$ trajectories. The red curves denote $w(\bm{c}^{**},\cdot)+\xi^{**}: [0,T]\rightarrow \mathbb{R}$ and $w(\bm{c}^{**},\cdot)-\xi^{**}: [0,T]\rightarrow \mathbb{R}$ respectively, where $w(\bm{c}^{**},\cdot)$ is the model of degree 5. The blue curves denote $w(\bm{c}^{**},\cdot)+\xi^{**}: [0,T]\rightarrow \mathbb{R}$ and $w(\bm{c}^{**},\cdot)-\xi^{**}: [0,T]\rightarrow \mathbb{R}$ respectively, where $w(\bm{c}^{**},\cdot)$ is the model of degree 2. \oomit{Blow: The blue points denote the extracted $10^4$ inputs in the $x_1-x_2$ plane.} }
     \label{fig43}
\end{figure}
\end{example}

\begin{example}
\label{odehigh}
To demonstrate the applicability of our approach to higher dimensional systems, we consider a scalable system of the form \eqref{bb} with $T=2$, $\mathcal{X}_0=[0.5,0.6]^{101}$ and $\mathtt{Uns}=\{y\in \mathbb{R}\mid y\geq 3.0\}$, describing the time evolution of the state $x_1$ in an ordinary differential equation \cite{ratschan2017}: 
\begin{equation*}
\begin{cases}
\dot{x}_1(t)=1+\frac{1}{l}(\sum_{i=1}^l x_{i+1}(t)+x_{i+2}(t)),\\
\dot{x}_2(t)=x_3(t),\dot{x}_3(t)=-10\sin x_2(t)-x_2(t)\\
\ldots\\
\dot{x}_{2l}(t)=x_{2l+1}(t),\dot{x}_{2l+1}(t)=-10\sin x_{2l}(t)-x_2(t)
\end{cases}
\end{equation*}
where $l=50$. 

Let $\epsilon_1=0.2$, $\beta_1=10^{-10}$, $\epsilon_2=0.2$ and $\beta_2=10^{-10}$. In this example we compute two polynomial models of degree 2 and 4 to illustrate our method.

1). We use $M=271$, $N=271$ and a polynomial $w(\bm{c},t)$ of degree $2$ as a mathematical model, which is input-independent, to perform computations. Note that the number $k+1$ of decision variables in \eqref{lp1} is $4$ and consequently $M\geq 271$ and $N\geq 271$ according to Theorem \ref{conclusion3}. Via solving \eqref{lp1} with $U_c=U_{\xi}=100$, we obtain that $\xi^{**}=0.36$. The computation time is $398.23$ seconds. According to Theorem \ref{conclusion3}, we have that with at least $1-10^{-10}$ confidence, the probability measure of inputs in $\mathcal{X}_0$ such that with confidence of at least $1-10^{-10}$, 
\[y_{\bm{x}_0}(t)\in [z_{\bm{x}_{0}}(t)-0.36, z_{\bm{x}_{0}}(t)+0.36]\]
for all $t\in [0,2]$ but at most a fraction $0.2$, is larger than $0.8$, where $z_{\bm{x}_{0}}(\cdot):[0,T] \rightarrow \mathbb{R}$ is the trajectory of the mathematical model $z(t)=w(\bm{c}^{**},t)$. The reachability analysis is illustrated in Fig. \ref{51}. Like Example \ref{nine}, within the Monte-Carlo testing framework, we also extract $10^4$ inputs $(\bm{x}'_{i,0})_{i=1}^{10^4}$ to verify the above conclusion, and obtain that the ratio of $10^4$ inputs  such that 
$y_{\bm{x}'_{i,0}}(j\Delta t)\in [z_{\bm{x}'_{i,0}}(j\Delta t)-0.36, z_{\bm{x}'_{0}}(j\Delta t)+0.36]$
for all $j\in \{0,\ldots,10^5\}$ is equal to $98.07\%$, where $\Delta t=\frac{2}{10^5}$ and $i=1,\ldots,10^4$.

Since $[z_{\bm{x}_{0}}(t)-0.36, z_{\bm{x}_{0}}(t)+0.36] \cap \mathtt{Uns}=\emptyset$ for $t\in [0,2]$ and $\bm{x}_0\in \mathcal{X}_0$, we have that with at least $1-10^{-10}$ confidence, the probability measure of inputs in $\mathcal{X}_0$ such that the amount of time the system \eqref{bb} with each of them spends inside $\mathtt{Uns}$ does not exceed $0.4$ with at least $1-10^{-10}$ confidence, is larger than $0.8$. 

\begin{figure}
\centering
   \includegraphics[width=3.05in,height=1.3in]{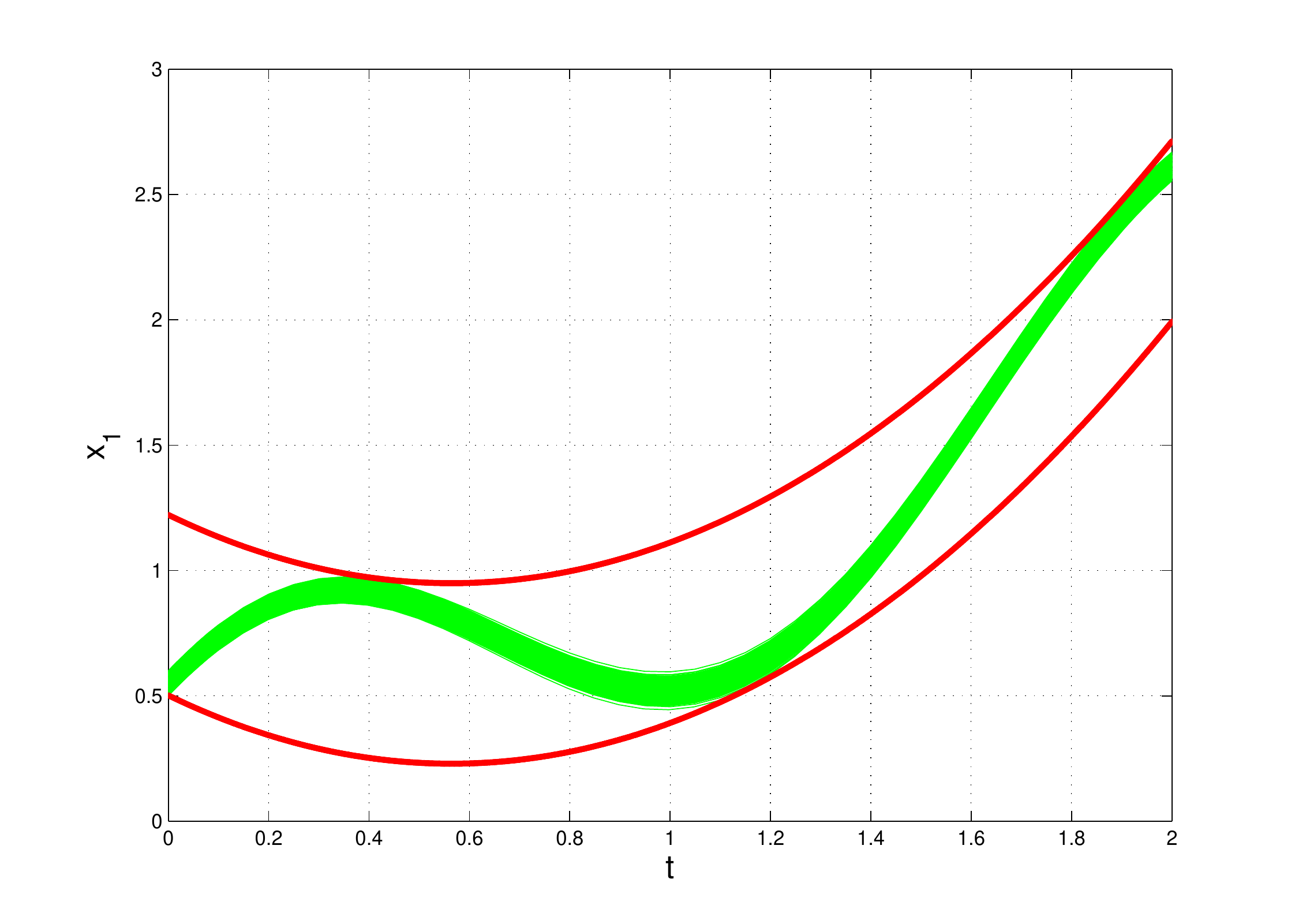}
  \caption{An illustration of all trajectories reachability for Example \ref{odehigh} with a polynomial model of degree 2. The green curves denote the extracted $271$ trajectories. The red curves denote $w(\bm{c}^{**},\cdot)-\xi^{**}:[0,T]\rightarrow \mathbb{R}$ and $w(\bm{c}^{**},\cdot)+\xi^{**}:[0,T]\rightarrow \mathbb{R}$ respectively. \oomit{The blue points denote the extracted $271$ inputs in the $x_1-x_2$ plane.} }
  \label{51}
  \end{figure}
  
2). We use $M=291$, $N=291$ and a polynomial $w(\bm{c},t)$ of degree $4$ as a mathematical model, which is input-independent, to perform computations. Note that the number $k+1$ of decision variables in \eqref{lp1} is $6$ and consequently $M\geq 291$ and $N\geq 291$ according to Theorem \ref{conclusion3}. Via solving \eqref{lp1} with $U_c=U_{\xi}=100$, we obtain that $\xi^{**}=0.12$. The computation time is $398.23$ seconds.  According to Theorem \ref{conclusion3}, we have that with at least $1-10^{-10}$ confidence, the probability measure of inputs in $\mathcal{X}_0$ such that with confidence of at least $1-10^{-10}$, 
\[y_{\bm{x}_0}(t)\in [z_{\bm{x}_{0}}(t)-0.12, z_{\bm{x}_{0}}(t)+0.12]\]
for all $t\in [0,2]$ but at most a fraction $0.2$, is larger than $0.8$, where $z_{\bm{x}_{0}}(\cdot):[0,T] \rightarrow \mathbb{R}$ is the trajectory of the mathematical model $z(t)=w(\bm{c}^{**},t)$. The reachability analysis is illustrated in Fig. \ref{52}. Also, within the Monte-Carlo testing framework we use the $10^4$ inputs $(\bm{x}'_{i,0})_{i=1}^{10^4}$ in the first case to verify the above conclusion, and obtain that the ratio of $10^4$ inputs such that 
$y_{\bm{x}'_{i,0}}(j\Delta t)\in [z_{\bm{x}'_{i,0}}(j\Delta t)-0.12, z_{\bm{x}'_{0}}(j\Delta t)+0.12]$
for all $j\in \{0,\ldots,10^5\}$ is equal to $1$, where $\Delta t=\frac{2}{10^5}$ and $i=1,\ldots,10^4$.

Similar to the first case, we have that with at least $1-10^{-10}$ confidence, the probability measure of inputs in $\mathcal{X}_0$ such that the amount of time the system \eqref{bb} with each of them spends inside the unsafe set $\mathtt{Uns}$ does not exceed $0.4$ with confidence of at least $1-10^{-10}$, is larger than $0.8$.

Like Example \ref{nine}, by comparing the results in Fig. \ref{53} for the above two cases with the same {\rm PAC} guarantees, i.e., $\epsilon_1$, $\epsilon_2$, $\beta_1$ and $\beta_2$ are the same, we also obtain that polynomial models of higher degree could capture the internal dynamics of the system \eqref{bb} more exactly, but with more computation time.

\begin{figure}
\centering
   \includegraphics[width=3.05in,height=1.2in]{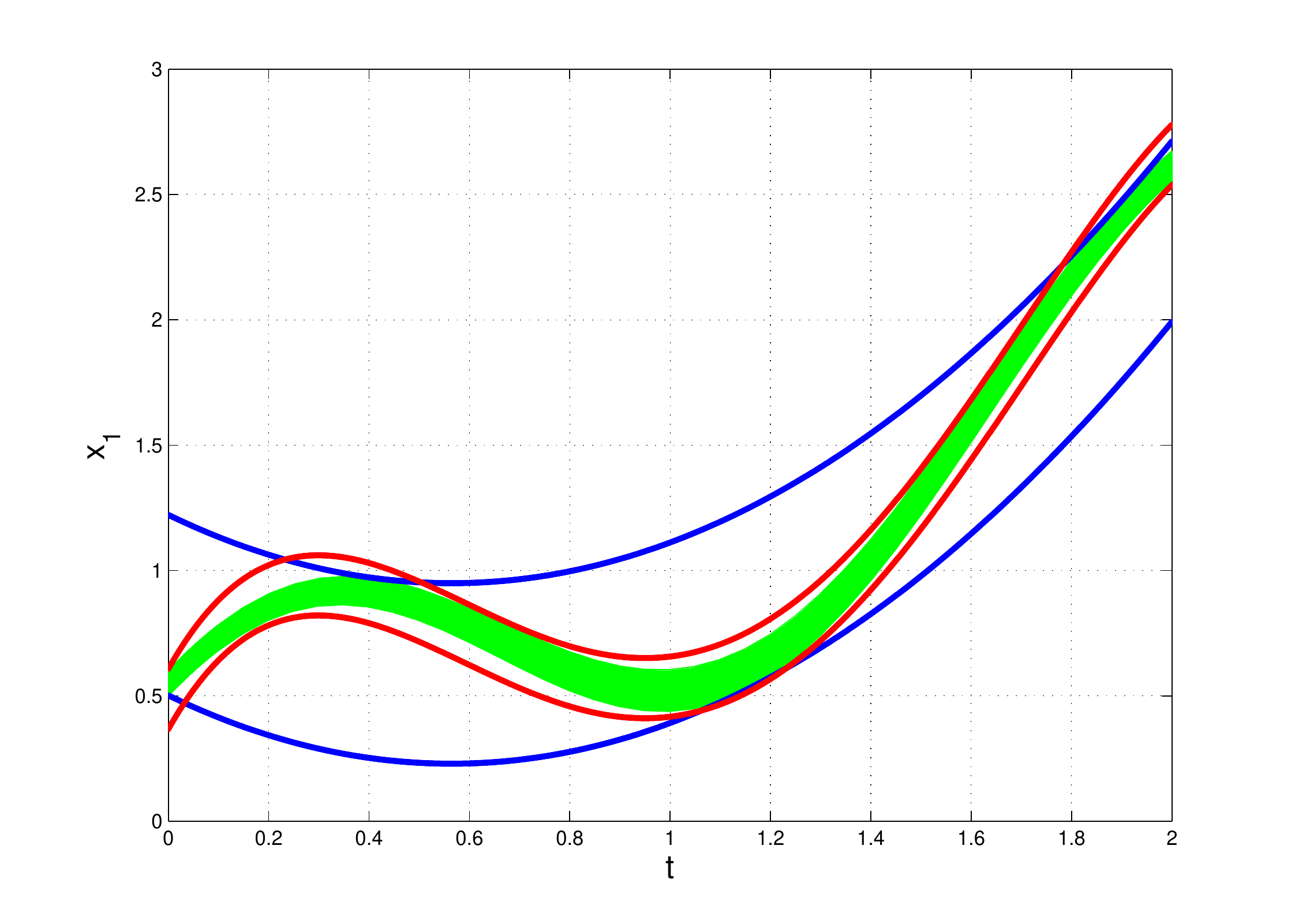}
  \caption{An illustration of trajectories reachability for Example \ref{odehigh} with a polynomial model of degree 4. The green curves denote the extracted $291$ trajectories. The red curves denote $w(\bm{c}^{**},\cdot)-\xi^{**}:[0,T]\rightarrow \mathbb{R}$ and $w(\bm{c}^{**},\cdot)+\xi^{**}:[0,T]\rightarrow \mathbb{R}$ respectively. \oomit{The blue points denote the extracted $291$ inputs in the $x_1-x_2$ plane.} }
  \label{52}
  \end{figure}
  \begin{figure}
  \centering
   \includegraphics[width=3.05in,height=1.2in]{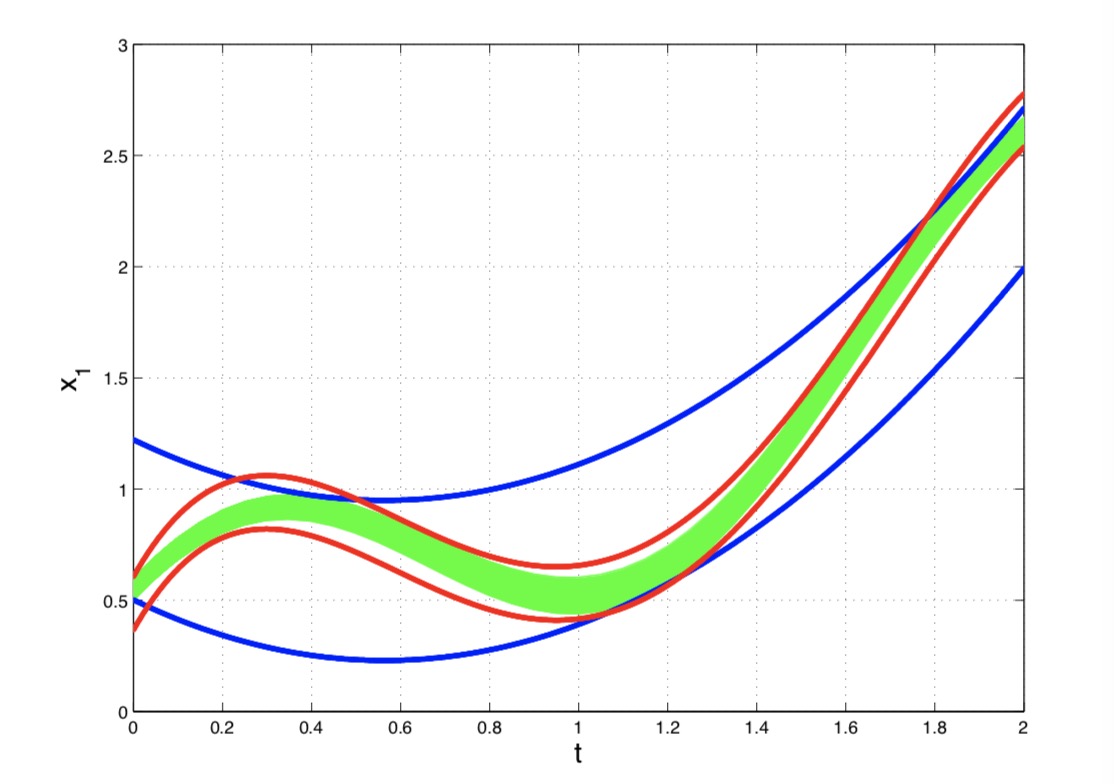}
     \caption{An illustration of Monte Carlo validation for Example \ref{odehigh}. The green curves denote the $10^4$ trajectories. The red curves denote $w(\bm{c}^{**},\cdot)+\xi^{**}: [0,T]\rightarrow \mathbb{R}$ and $w(\bm{c}^{**},\cdot)-\xi^{**}: [0,T]\rightarrow \mathbb{R}$ respectively, where $w(\bm{c}^{**},\cdot)$ is the {\rm PAC} model of degree 4. The blue curves denote $w(\bm{c}^{**},\cdot)+\xi^{**}: [0,T]\rightarrow \mathbb{R}$ and $w(\bm{c}^{**},\cdot)-\xi^{**}: [0,T]\rightarrow \mathbb{R}$ respectively, where $w(\bm{c}^{**},\cdot)$ is the {\rm PAC} model of degree 2. \oomit{The blue points denote the $10^4$ inputs in the $x_1-x_2$ plane. }}
     \label{53}
\end{figure}
\end{example}

\begin{example}
\label{delay1}
In this example we show a strategy to overcome the issue of solving large-scale linear programs based on a black-box system of the form \eqref{bb} which describes the time evolution of the state $x_1$ in the two-dimensional delay differential equation
\begin{equation*}
\left\{
\begin{aligned}
&\dot{x}_1(t)=ax_1(t)(1-\frac{x_1(t)}{m})+bx_1(t)x_2(t)\\
&\dot{x}_2(t)=cx_2(t)+dx_1(t-\tau)x_2(t-\tau)
\end{aligned}
\right.
\end{equation*}
where $\tau=0.1$, $a=0.25$, $m=200$, $b=-0.01$, $c=-1.00$ and $d=0.01$. The delay differential equation was a model for predator-prey populations. 

Assume that $T=10$, the initial condition $\bm{x}(t)$ over $t\in [-0.1,0]$ is a constant vector falling within $\mathcal{X}_0=\{(x_1,x_2)\mid (x_1+5)^2+(x_2+5)^2\leq 1\}$ and $\mathtt{Uns}=\{y\mid y\geq 40\}$.

Let $\epsilon_1=0.1$, $\beta_1=10^{-10}$, $\epsilon_2=0.1$ and $\beta_2=10^{-10}$. In this example we first use input-dependent polynomial models of degree $4$ to illustrate this strategy, and then use input-independent polynomial models of degree $4$ to illustrate it.

1).  Input-dependent Models: If a generic  polynomial input-dependent model template of degree 4, which is formed by choosing all monomials of degree up to 4 as the basis polynomials,  is employed, the number $k+1$ of decision variables in \eqref{lp1} is $36$ and consequently $M\geq 1181$ and $N\geq 1181$ according to Theorem \ref{conclusion3}. This leads to a large-scale linear program, producing heavy computational burden. As a result, we did not obtain results within two hours via solving this large-scale linear program.

Our strategy for avoiding large-scale linear programs is as follows: a small family of datum is first employed to compute an initial estimate of the coefficients $\bm{c}$, and then determine the values of some coefficients based on the computed $\bm{c}$ and leave the remaining ones unknown, reducing the number of decision variables in \eqref{lp1} and thus the size of the resulting linear program.

In the experiment we first solve the linear program \eqref{lp1} with $M=50$ and $N=50$ to obtain a model $w'(\bm{c}^{**},\bm{x},t)$ with the computation time of 1.82 seconds, and then use the computed $w'(\bm{c}^{**},\bm{x},t)$ to perform computations on the linear program \eqref{lp1} with $M=N=481$ and $U_c=U_{\xi}=100$. Note that the number $k+1$ of decision variables in \eqref{lp1} becomes $1$ in this setting and consequently $M\geq 481$ and $N\geq 481$ according to Theorem \ref{conclusion3}. Via solving \eqref{lp1} with $U_c=U_{\xi}=100$, we obtain that $\xi^{**}=1.49$ with the computation time of 268.67 seconds. The reachability analysis is illustrated in Fig. \ref{61}. Therefore, according to Theorem \ref{conclusion3}, we conclude that with at least $1-10^{-10}$ confidence, the probability measure of inputs in $\mathcal{X}_0$ such that with confidence of at least $1-10^{-10}$, $y_{\bm{x}_0}(t)\in [z_{\bm{x}_{0}}(t)-1.49, z_{\bm{x}_{0}}(t)+1.49]$
for all $t\in [0,10]$ but at most a fraction $0.1$, is larger than $0.9$, where $z_{\bm{x}_{0}}(\cdot):[0,T] \rightarrow \mathbb{R}$ is the trajectory of the mathematical model $z(t)=w'(\bm{c}^{**},\bm{x},t)$. Also, within the Monte-Carlo framework, we extract $10^4$ inputs $(\bm{x}'_{i,0})_{i=1}^{10^4}$ to verify the above conclusion, and obtain that the ratio of $10^4$ inputs such that
$y_{\bm{x}'_{i,0}}(j\Delta t)\in [z_{\bm{x}'_{i,0}}(j\Delta t)-1.49, z_{\bm{x}'_{i,0}}(j\Delta t)+1.49]$
for all $j\in \{0,\ldots,10^6\}$ is $100\%$, where $\Delta t=10^{-5}$.

Since $[z_{\bm{x}_{0}}(t)-1.49, z_{\bm{x}_{0}}(t)+1.49] \cap \mathtt{Uns}=\emptyset$ for $t\in [0,10]$ and $\bm{x}_0\in \mathcal{X}_0$, we have that with at least $1-10^{-10}$ confidence, the probability measure of inputs in $\mathcal{X}_0$ such that the amount of time the system \eqref{bb} with each of them spends inside $\mathtt{Uns}$ does not exceed $1$ with confidence of at least $1-10^{-10}$, is larger than $0.9$. 

2). Input-independent  Models:  If an input-independent polynomial  template of degree $4$ is used to perform computations, the number $k+1$ of decision variables in \eqref{lp1} is $6$ and consequently $M\geq 581$ and $N\geq 581$ according to Theorem \ref{conclusion3}. Via solving the linear program \eqref{lp1} with $M=N=581$ and $U_c=U_{\xi}=100$, we obtain $\xi^{**}=24.84$ with the computation time of 6634.51 seconds. The reachability analysis is illustrated in Fig. \ref{62}.

We also adopt the strategy presented in the above case for reducing the computation cost.  We first solve the linear program \eqref{lp1} with $M=N=50$ and $U_c=U_{\xi}=100$ to obtain a $w'(\bm{c}^{**},t)$ with the computation time of 1.65 seconds, and then use the computed $w'(\bm{c}^{**},t)$ to perform computations on the linear program \eqref{lp1} with $M=N=481$ and $U_c=U_{\xi}=100$. Note that the number $k+1$ of decision variables in \eqref{lp1} becomes $1$ in this setting and consequently $M\geq 481$ and $N\geq 481$ according to Theorem \ref{conclusion3}. Via solving \eqref{lp1} with $U_c=U_{\xi}=100$, we obtain that $\xi^{**}=25.96$ with the computation time of 71.09 seconds. The reachability analysis is illustrated in Fig. \ref{62} as well.
The safety guarantee is the same with the case of using input-dependent models. Similarly, within the Monte-Carlo framework, we use the $10^4$ inputs $(\bm{x}'_{i,0})_{i=1}^{10^4}$ in the first case to verify the above conclusion, and obtain that the ratio of $10^4$ inputs such that
$y_{\bm{x}'_{i,0}}(j\Delta t)\in [z_{\bm{x}'_{i,0}}(j\Delta t)-25.96, z_{\bm{x}'_{i,0}}(j\Delta t)+25.96]$
for all $j\in \{0,\ldots,10^6\}$ is equal to $100\%$, where $\Delta t=10^{-5}$.

Via comparing the results in Fig. \ref{61} and \ref{62} for the above two cases with the same {\rm PAC} guarantees, i.e., $\epsilon_1$, $\epsilon_2$, $\beta_1$ and $\beta_2$ are the same, we conclude that input-dependent polynomial models could capture the internal dynamics of the system \eqref{bb} more exactly than input-independent ones, but also with more computation cost.

\begin{figure}
\centering
   \includegraphics[width=3.05in,height=1.2in]{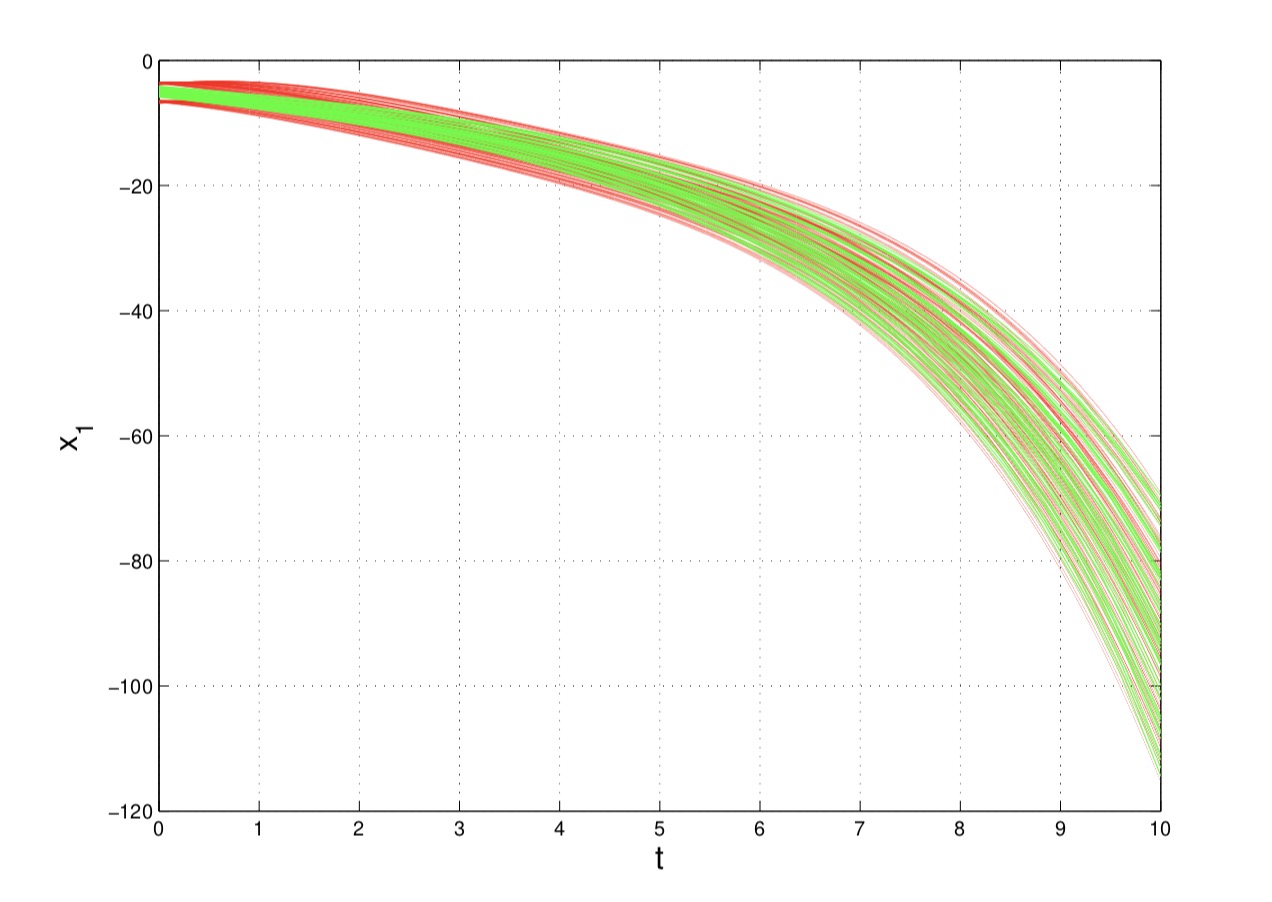}
  \caption{An illustration of trajectories reachability for Example \ref{delay1} with the input-dependent model $w'(\bm{c}^*,\bm{x},t)$. The green curves denote some extracted trajectories. The red curves denote the corresponding $w'(\bm{c}^{**},\bm{x},\cdot)-\xi^{**}:[0,T]\rightarrow \mathbb{R}$ and $w'(\bm{c}^{**},\bm{x},\cdot)+\xi^{**}:[0,T]\rightarrow \mathbb{R}$ respectively. \oomit{The blue points denote the extracted $291$ inputs in the $x_1-x_2$ plane.} }
  \label{61}
  \end{figure}
  
  \begin{figure}
\centering
   \includegraphics[width=3.05in,height=1.2in]{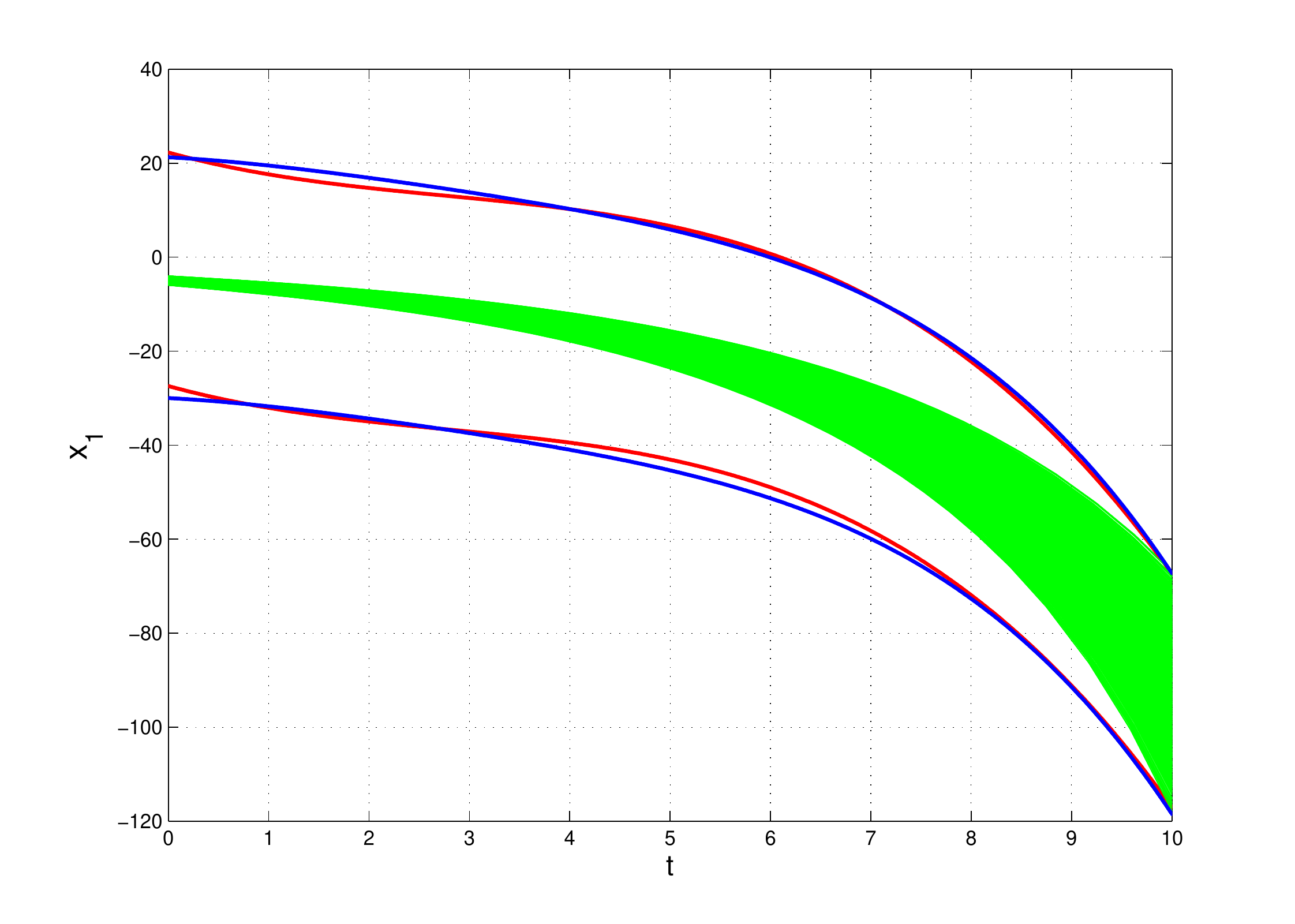}
  \caption{An illustration of trajectories reachability for Example \ref{delay1} with the input-independent model $w(\bm{c}^*,t)$. The green curves denote some extracted trajectories. The red curves denote the corresponding $w(\bm{c}^{**}, \cdot)-\xi^{**}:[0,T]\rightarrow \mathbb{R}$ and $w(\bm{c}^{**},\cdot)+\xi^{**}:[0,T]\rightarrow \mathbb{R}$ with $\xi^{**}=24.84$ respectively. The blue curves denote the corresponding $w'(\bm{c}^{**}, \cdot)-\xi^{**}:[0,T]\rightarrow \mathbb{R}$ and $w'(\bm{c}^{**},\cdot)+\xi^{**}:[0,T]\rightarrow \mathbb{R}$ with $\xi^{**}=25.96$ respectively. }
  \label{62}
  \end{figure}  
\end{example}

\section{Conclusion}
\label{conclusion}
In this paper we proposed a novel PAC model checking approach for finite-time safety verification of black-box continuous-time dynamical systems, which are represented by observed datum, within the framework of PAC learning. In this approach, a PAC model of the system was computed such that the time-evolving trajectories of the black-box dynamical system over finite-time horizons fall within the range of the PAC model plus a bounded interval with error probabilities and confidence levels, thus facilitating the formal characterization of the satisfiability of safety requirements. Both the PAC model and the bounded interval were obtained via scenario optimization, which finally boil down to a linear program. Three examples demonstrated the performance of our approach.

In the future we would extend our method to safety verification of black-box systems, whose internal mechanisms are described by hybrid dynamical systems that exhibit both continuous and discrete dynamic behavior. Also, we would like to extend our method for safety verification of black-box systems with noise measurements and inputs.

\bibliographystyle{abbrv}
\bibliography{reference}

\begin{thebibliography}{10}

\bibitem{aichernig2019}
B.~K. Aichernig and M.~Tappler.
\newblock Probabilistic black-box reachability checking (extended version).
\newblock {\em Formal methods in system design}, 54(3):416--448, 2019.

\bibitem{ashok2019pac}
P.~Ashok, J.~K{\v{r}}et{\'\i}nsk{\`y}, and M.~Weininger.
\newblock Pac statistical model checking for markov decision processes and
  stochastic games.
\newblock In {\em CAV'19}, pages 497--519. Springer, 2019.

\bibitem{boyer2013plasma}
B.~Boyer, K.~Corre, A.~Legay, and S.~Sedwards.
\newblock Plasma-lab: A flexible, distributable statistical model checking
  library.
\newblock In {\em QEST'13}, pages 160--164. Springer, 2013.

\bibitem{brazdil2014}
T.~Br{\'a}zdil, K.~Chatterjee, M.~Chmelik, V.~Forejt,
  J.~K{\v{r}}et{\'\i}nsk{\`y}, M.~Kwiatkowska, D.~Parker, and M.~Ujma.
\newblock Verification of markov decision processes using learning algorithms.
\newblock In {\em ATVA'14}, pages 98--114. Springer, 2014.

\bibitem{calafiore2006}
G.~C. Calafiore and M.~C. Campi.
\newblock The scenario approach to robust control design.
\newblock {\em IEEE Transactions on Automatic Control}, 51(5):742--753, 2006.

\bibitem{campi2009}
M.~C. Campi, S.~Garatti, and M.~Prandini.
\newblock The scenario approach for systems and control design.
\newblock {\em Annual Reviews in Control}, 33(2):149--157, 2009.

\bibitem{castelvecchi2016}
D.~Castelvecchi.
\newblock Can we open the black box of ai?
\newblock {\em Nature News}, 538(7623):20, 2016.

\bibitem{chen2015}
X.~Chen.
\newblock {\em Reachability Analysis of Non-Linear Hybrid Systems Using Taylor
  Models}.
\newblock PhD thesis, Fachgruppe Informatik, RWTH Aachen University, 2015.

\bibitem{ChenAS13}
X.~Chen, E.~{\'{A}}brah{\'{a}}m, and S.~Sankaranarayanan.
\newblock Flow*: An analyzer for non-linear hybrid systems.
\newblock In {\em CAV'13}, pages 258--263. Springer, 2013.

\bibitem{chen2016pac}
Y.-F. Chen, C.~Hsieh, O.~Leng{\'a}l, T.-J. Lii, M.-H. Tsai, B.-Y. Wang, and
  F.~Wang.
\newblock Pac learning-based verification and model synthesis.
\newblock In {\em ICSE'16}, pages 714--724. IEEE, 2016.

\bibitem{clarke2008}
E.~M. Clarke, J.~R. Faeder, C.~J. Langmead, L.~A. Harris, S.~K. Jha, and
  A.~Legay.
\newblock Statistical model checking in biolab: Applications to the automated
  analysis of t-cell receptor signaling pathway.
\newblock In {\em CMSB'08}, pages 231--250. Springer, 2008.

\bibitem{clarke1994model}
E.~M. Clarke, O.~Grumberg, and D.~E. Long.
\newblock Model checking and abstraction.
\newblock {\em ACM transactions on Programming Languages and Systems (TOPLAS)},
  16(5):1512--1542, 1994.

\bibitem{clarke2011}
E.~M. Clarke and P.~Zuliani.
\newblock Statistical model checking for cyber-physical systems.
\newblock In {\em ATVA'11}, pages 1--12. Springer, 2011.

\bibitem{david2015uppaal}
A.~David, K.~G. Larsen, A.~Legay, M.~Miku{\v{c}}ionis, and D.~B. Poulsen.
\newblock Uppaal smc tutorial.
\newblock {\em International Journal on Software Tools for Technology
  Transfer}, 17(4):397--415, 2015.

\bibitem{DuggiralaMVP15}
P.~S. Duggirala, S.~Mitra, M.~Viswanathan, and M.~Potok.
\newblock {C2E2:} {A} verification tool for stateflow models.
\newblock In {\em TACAS'15}, pages 68--82. Springer, 2015.

\bibitem{FanQM017}
C.~Fan, B.~Qi, S.~Mitra, and M.~Viswanathan.
\newblock Dryvr: Data-driven verification and compositional reasoning for
  automotive systems.
\newblock In {\em CAV'17}, pages 441--461. Springer, 2017.

\bibitem{fan2020parameter}
C.~Fan, X.~Qin, and J.~Deshmukh.
\newblock Parameter searching and partition with probabilistic coverage
  guarantees.
\newblock {\em arXiv preprint arXiv:2004.00279}, 2020.

\bibitem{franzle2015multi}
M.~Fr{\"a}nzle, S.~Gerwinn, P.~Kr{\"o}ger, A.~Abate, and J.-P. Katoen.
\newblock Multi-objective parameter synthesis in probabilistic hybrid systems.
\newblock In {\em FORMATS'15}, pages 93--107. Springer, 2015.

\bibitem{fu2014probably}
J.~Fu and U.~Topcu.
\newblock Probably approximately correct mdp learning and control with temporal
  logic constraints.
\newblock {\em arXiv preprint arXiv:1404.7073}, 2014.

\bibitem{grosu2005monte}
R.~Grosu and S.~A. Smolka.
\newblock Monte carlo model checking.
\newblock In {\em TACAS'05}, pages 271--286. Springer, 2005.

\bibitem{horst2013global}
R.~Horst and H.~Tuy.
\newblock {\em Global optimization: Deterministic approaches}.
\newblock Springer Science \& Business Media, 2013.

\bibitem{immler2015}
F.~Immler.
\newblock Verified reachability analysis of continuous systems.
\newblock In {\em TACAS'15}, pages 37--51. Springer, 2015.

\bibitem{jha2009bayesian}
S.~K. Jha, E.~M. Clarke, C.~J. Langmead, A.~Legay, A.~Platzer, and P.~Zuliani.
\newblock A bayesian approach to model checking biological systems.
\newblock In {\em CMSB'09}, pages 218--234. Springer, 2009.

\bibitem{katoen2011ins}
J.-P. Katoen, I.~S. Zapreev, E.~M. Hahn, H.~Hermanns, and D.~N. Jansen.
\newblock The ins and outs of the probabilistic model checker mrmc.
\newblock {\em Performance evaluation}, 68(2):90--104, 2011.

\bibitem{lee2008cyber}
E.~A. Lee.
\newblock Cyber physical systems: Design challenges.
\newblock In {\em ISORC'08}, pages 363--369. IEEE, 2008.

\bibitem{mao2011learning}
H.~Mao, Y.~Chen, M.~Jaeger, T.~D. Nielsen, K.~G. Larsen, and B.~Nielsen.
\newblock Learning probabilistic automata for model checking.
\newblock In {\em QEST'11}, pages 111--120. IEEE, 2011.

\bibitem{mao2012learning}
H.~Mao, Y.~Chen, M.~Jaeger, T.~D. Nielsen, K.~G. Larsen, and B.~Nielsen.
\newblock Learning markov decision processes for model checking.
\newblock {\em arXiv preprint arXiv:1212.3873}, 2012.

\bibitem{mao2016learning}
H.~Mao, Y.~Chen, M.~Jaeger, T.~D. Nielsen, K.~G. Larsen, and B.~Nielsen.
\newblock Learning deterministic probabilistic automata from a model checking
  perspective.
\newblock {\em Machine Learning}, 105(2):255--299, 2016.

\bibitem{nouri2014faster}
A.~Nouri, B.~Raman, M.~Bozga, A.~Legay, and S.~Bensalem.
\newblock Faster statistical model checking by means of abstraction and
  learning.
\newblock In {\em RV'14}, pages 340--355. Springer, 2014.

\bibitem{park2019pac}
S.~Park, O.~Bastani, N.~Matni, and I.~Lee.
\newblock Pac confidence sets for deep neural networks via calibrated
  prediction.
\newblock {\em arXiv preprint arXiv:2001.00106}, 2019.

\bibitem{peled1999black}
D.~Peled, M.~Y. Vardi, and M.~Yannakakis.
\newblock Black box checking.
\newblock In {\em Formal Methods for Protocol Engineering and Distributed
  Systems}, pages 225--240. Springer, 1999.

\bibitem{rajkumar2010cyber}
R.~Rajkumar, I.~Lee, L.~Sha, and J.~Stankovic.
\newblock Cyber-physical systems: the next computing revolution.
\newblock In {\em Design Automation Conference}, pages 731--736. IEEE, 2010.

\bibitem{ratschan2017}
S.~Ratschan.
\newblock Simulation based computation of certificates for safety of dynamical
  systems.
\newblock In {\em FORMATS'17}, pages 303--317. Springer, 2017.

\bibitem{ReijsbergenBSH15}
D.~Reijsbergen, P.~de~Boer, W.~R.~W. Scheinhardt, and B.~R. Haverkort.
\newblock On hypothesis testing for statistical model checking.
\newblock {\em Int. J. Softw. Tools Technol. Transf.}, 17(4):377--395, 2015.

\bibitem{sen2004statistical}
K.~Sen, M.~Viswanathan, and G.~Agha.
\newblock Statistical model checking of black-box probabilistic systems.
\newblock In {\em CAV'04}, pages 202--215. Springer, 2004.

\bibitem{SenVA05}
K.~Sen, M.~Viswanathan, and G.~A. Agha.
\newblock {VESTA:} {A} statistical model-checker and analyzer for probabilistic
  systems.
\newblock In {\em QEST'05}, pages 251--252. {IEEE} Computer Society, 2005.

\bibitem{shalev2014}
S.~Shalev-Shwartz and S.~Ben-David.
\newblock {\em Understanding machine learning: From theory to algorithms}.
\newblock Cambridge university press, 2014.

\bibitem{valiant2013}
L.~Valiant.
\newblock {\em Probably Approximately Correct: Nature{\~O}s Algorithms for
  Learning and Prospering in a Complex World}.
\newblock Basic Books (AZ), 2013.

\bibitem{van1926}
B.~Van~der Pol.
\newblock Lxxxviii. on “relaxation-oscillations”.
\newblock {\em The London, Edinburgh, and Dublin Philosophical Magazine and
  Journal of Science}, 2(11):978--992, 1926.

\bibitem{waga2020}
M.~Waga.
\newblock Falsification of cyber-physical systems with robustness-guided
  black-box checking.
\newblock In {\em HSCC'20}, pages 1--13, 2020.

\bibitem{Wald1945}
A.~Wald.
\newblock Sequential tests of statistical hypotheses.
\newblock {\em The Annals of Mathematical Statistics}, 16(2):117--186, 1945.

\bibitem{xue2019probably}
B.~Xue, M.~Fr{\"a}nzle, H.~Zhao, N.~Zhan, and A.~Easwaran.
\newblock Probably approximate safety verification of hybrid dynamical systems.
\newblock In {\em ICFEM'19}, pages 236--252. Springer, 2019.

\bibitem{xue2019safe}
B.~Xue, Y.~Liu, L.~Ma, X.~Zhang, M.~Sun, and X.~Xie.
\newblock Safe inputs approximation for black-box systems.
\newblock In {\em ICECCS'19}, pages 180--189. IEEE, 2019.

\bibitem{9023360}
B.~{Xue}, Q.~{Wang}, S.~{Feng}, and N.~{Zhan}.
\newblock Over- and under-approximating reach sets for perturbed delay
  differential equations.
\newblock {\em IEEE Transactions on Automatic Control}, pages 1--1, 2020.

\bibitem{younes2005}
H.~L. Younes.
\newblock Probabilistic verification for “black-box” systems.
\newblock In {\em CAV'05}, pages 253--265. Springer, 2005.

\bibitem{younes2005ymer}
H.~L. Younes.
\newblock Ymer: A statistical model checker.
\newblock In {\em CAV'05}, pages 429--433. Springer, 2005.

\bibitem{younes2002}
H.~L. Younes and R.~G. Simmons.
\newblock Probabilistic verification of discrete event systems using acceptance
  sampling.
\newblock In {\em CAV'02}, pages 223--235. Springer, 2002.

\bibitem{zuliani2013}
P.~Zuliani, A.~Platzer, and E.~M. Clarke.
\newblock Bayesian statistical model checking with application to
  stateflow/simulink verification.
\newblock {\em Formal Methods in System Design}, 43(2):338--367, 2013.

\end{thebibliography}
\end{document}